\pdfoutput=1
\documentclass[journal, draftclsnofoot, onecolumn, 12pt]{IEEEtran}
\usepackage{mathtools} 
\usepackage{amsthm, amssymb} 
\usepackage{bm} 
\usepackage{microtype} 
\usepackage[doi = false, isbn = false, url = false, style = ieee, citecounter = context]{biblatex} 
\usepackage{orcidlink} 
\usepackage{hyperref} 
\usepackage[nameinlink]{cleveref} 

\usetikzlibrary{shapes, calc, decorations.pathreplacing}
\hypersetup{hidelinks, pdfinfo = {Title = {Fundamental Limits of Combinatorial Multi-Access Caching}, Author = {Federico Brunero and Petros Elia}, Subject = {Information-Theoretic Caching}, Keywords = {Coded Caching, Combinatorial Topology, Index Coding, Information-Theoretic Converse, Multi-Access Coded Caching (MACC)}}}
\allowdisplaybreaks

\theoremstyle{plain}
\newtheorem{theorem}{Theorem}
\newtheorem{lemma}{Lemma}
\newtheorem{corollary}{Corollary}[theorem]

\theoremstyle{definition}
\newtheorem{example}{Example}

\theoremstyle{remark}
\newtheorem{remark}{Remark}

\newcommand{\refappendix}[1]{\hyperref[#1]{Appendix~\ref*{#1}}}

\bibliography{references.bib}
\AtEveryBibitem{\clearlist{publisher}} 

\title{Fundamental Limits of Combinatorial Multi-Access Caching}
\author{Federico Brunero\textsuperscript{\orcidlink{0000-0002-6980-3827}} and Petros Elia\textsuperscript{\orcidlink{0000-0002-3531-120X}}%
  \thanks{%
    This work was supported by the European Research Council (ERC) through the EU Horizon 2020 Research and Innovation Program under Grant 725929 (Project DUALITY).

    The authors are with the Communication Systems Department at EURECOM, 450 Route des Chappes, 06410 Sophia Antipolis, France (email: brunero@eurecom.fr; elia@eurecom.fr).
  }
}

\begin{document}

\maketitle

\begin{abstract}
    This work identifies the fundamental limits of Multi-Access Coded Caching (MACC) where each user is connected to multiple caches in a manner that follows a generalized combinatorial topology. This topology stands out as it allows for unprecedented coding gains, even with very modest cache resources. First, we extend the setting and the scheme presented by Muralidhar \emph{et al.} to a much more general topology that supports both a much denser range of users and the coexistence of users connected to different numbers of caches, all while maintaining the astounding coding gains --- here proven to be exactly optimal --- associated with the combinatorial topology. This is achieved, for this generalized topology, with a novel information-theoretic converse that we present here, which establishes, together with the scheme, the exact optimal performance under the assumption of uncoded placement. We subsequently consider different connectivity ensembles, including the very general scenario of the entire ensemble of all possible network connectivities/topologies, where any user can be connected to any subset of caches arbitrarily. For these settings, we develop novel converse bounds on the optimal performance averaged over the ensemble's different connectivities, considering the additional realistic scenario that the connectivity at delivery time is entirely unknown during cache placement. This novel analysis of topological ensembles leaves open the possibility that currently-unknown topologies may yield even higher gains, a hypothesis that is part of the bigger question of which network topology yields the most caching gains.
\end{abstract}


\section{Introduction}

\IEEEPARstart{C}{oded} caching is a coding-based communication technique introduced in~\cite{Maddah-Ali2014FundamentalCaching} that allows for a significant reduction in the amount of data needed to be transferred from a centralized server to its cache-aided receiving users. The idea behind coded caching involves filling preemptively the caches of the users during the \emph{placement phase} without knowing the future requests, where such placement is done in a careful manner so that, once the requests of the users are revealed, the interference is reduced and thus the amount of bits to be transferred during the \emph{delivery phase} is minimized. As a result of the innovative cache placement and the subsequent clique-based delivery scheme, coded caching is able to provide sizeable gains. For instance, if we consider the standard single-stream broadcast channel model with $K$ receiving users each of which is able to store in its cache a fraction $\gamma$ of the main library, then the coding gain\footnote{In a nutshell, such coding gain represents simply the number of users to which a coded message is useful at the same time.} provided by the coded caching framework is equal to $K\gamma + 1$.

Since its original information-theoretic formulation, coded caching has been explored in several variations that involve the interplay between caching and PHY~\cite{Shariatpanahi2016Multi-ServerCaching, Tolli2020Multi-AntennaCaching, Shariatpanahi2018OnRegime, Zhang2017FundamentalFeedback, Lampiris2017Cache-aidedCSIT, MohammadiAmiri2018Cache-AidedChannels, Mohajer2020MISOSubpacketization, Bergel2018Cache-AidedSNR, Naderializadeh2017FundamentalManagement, Shariatpanahi2019Physical-LayerCaching, Cao2020DegreesNetworks}, caching and privacy~\cite{Engelmann2017ACaching, Yan2021FundamentalUsers, Wan2021OnDemands}, information-theoretic converses~\cite{Wan2020AnPlacement, Yu2018ThePrefetching}, and the critical bottleneck of subpacketization~\cite{Yan2017OnScheme, Tang2018CodedCodes, Krishnan2018CodedGraphs,Shangguan2018CentralizedApproach, Lampiris2018AddingSizes, Chittoor2021SubexponentialGeometry, Cheng2021LinearNetworks}, while more recently we saw interesting findings on coded caching under file popularity considerations~\cite{Niesen2017CodedDemands, Zhang2018CodedDistributions, Ji2017Order-OptimalDemands, Hachem2017CodedAccess, Saberali2020FullDemands} or in scenarios with heterogeneous user preferences~\cite{Brunero2021UnselfishCaching, Wan2021OnContent, Wang2019CodedProfiles, Zhang2020OnProfiles, Chang2019CodedCaching, Chang2020OnSets}, as well as on a variety of other scenarios~\cite{Ibrahim2020Device-to-DeviceSizes, Parrinello2020FundamentalPrefetching, Lampiris2020FullUsers, Wei2017NovelPrefetching, Joudeh2021FundamentalTraffic, Veld2020CachingGaussians, Daniel2020OptimizationCaching, Hassanzadeh2020Rate-MemorySources, Zhang2021DeepPopularity}.

\subsection{Multi-Access Coded Caching}\label{sec: MACC Introduction}

The original coded caching model in~\cite{Maddah-Ali2014FundamentalCaching} considered that each user has access to its own single dedicated cache. However, in several scenarios it is conceivable and perhaps more realistic that each cache serves more than one user, and that each user can connect to more than one cache. For instance, in dense cellular networks, the cache-aided Access Points (APs) could have overlapping coverage areas, hence allowing each user to connect to more than one AP. Even more realistically, in a wired setting where a central server wishes to communicate to multiple workers via a shared control channel, each worker could be assisted by multiple memory devices which are again shared among the workers. Such scenarios motivated the work in~\cite{Hachem2017CodedAccess} that introduced the extra dimension of having users that can now have access to multiple caches. In this setting --- in addition to the number of users $K$, the number of library files $N$ and the cache size of $M$ files --- a new  parameter $\lambda$ describes the number of caches that each user can access. This so-called Multi-Access Coded Caching (MACC) model introduced in~\cite{Hachem2017CodedAccess} involved $\Lambda$ users and $\Lambda$ caches, and involved a topology where each user is connected to $\lambda > 1$ consecutive caches in a cyclic wrap-around fashion as in~\Cref{fig: MACC Model Example}, such that each cache serves exactly $\lambda$ users. In the same work, the authors provided a caching-and-delivery procedure which guarantees in its centralized variant\footnote{The original work in~\cite{Hachem2017CodedAccess} proposed the coding scheme with decentralized (stochastic) cache placement. The centralized version can be easily obtained as also mentioned in~\cite{Reddy2020Rate-MemoryPlacement}.} a worst-case load of
\begin{equation}
    \frac{\Lambda(1 - \lambda \gamma)}{\Lambda\gamma + 1} = \frac{K(1 - \lambda \gamma)}{K\gamma + 1} 
\end{equation}
where $\gamma \coloneqq M/N$ is the fraction of the library that each cache is able to store. Such scheme takes advantage of the multi-access nature of the network and allows for an increase of the local caching gain from $\gamma$ to $\lambda\gamma$, without though being able to increase the coding gain, which remains fixed at the gain $\Lambda\gamma + 1$ that only corresponds to the gain in the dedicated cache scenario where $\lambda = 1$.

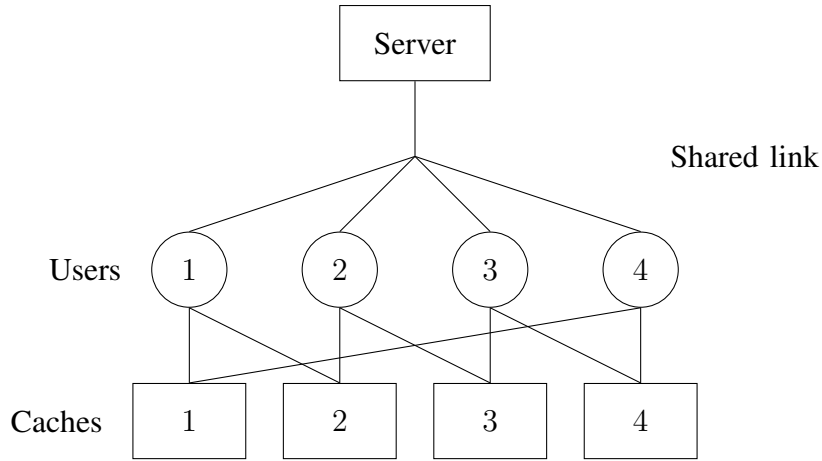
\begin{figure}[!htbp]
    \centering
    \tikzset{cnode/.style={draw, circle, minimum size = 1cm}}
    \tikzset{snode/.style={draw, rectangle, text centered}}
    \begin{tikzpicture}
        \node[snode, minimum height = 1cm, minimum width = 2cm](server){Server};
        \draw (server.south)--++(0, -1);
        \foreach \u/\n in {1/-3, 2/-1, 3/1, 4/3}{
            \draw ($(server.south) + (0, -1)$)--++(\n, -1) node[cnode, below](user\u){$\u$};
        }
        \node at (user1.west)[left, xshift = -0.25cm]{Users};
        \node at ($(server.south) + (3, -1)$)[right, xshift = 0.25cm]{Shared link};
        \foreach \c in {1, 2, 3, 4}{
            \node at ($(user\c.south) + (0, -1)$)[snode, below, minimum height = 1cm, minimum width = 1.5cm](cache\c){$\c$};
        }
        \foreach \c [evaluate = \c as \d using {int(mod(\c, 4)+1)}] in {1, 2, 3, 4}{
            \draw (user\c.south)--(cache\c.north);
            \draw (user\c.south)--(cache\d.north);
        }
        \node at (cache1.west)[left, xshift = -0.25cm]{Caches};
    \end{tikzpicture}
    \caption{MACC model where there are $\Lambda = 4$ caches and each user is connected to $\lambda = 2$ consecutive caches following a cyclic wrap-around topology.}
    \label{fig: MACC Model Example}
\end{figure}

Since the introduction of the aforementioned multi-access cyclic model, various works focused on the design of coding schemes that leverage the fact that each user has access to more cache space. The challenge is always to be able to achieve higher coding gains in a setting where the cache volume seen by a user must be shared among several users. One such work can be found in~\cite{Serbetci2019Multi-accessCache-redundancy}, which proposed a scheme that not only preserves the full local caching gain as in~\cite{Hachem2017CodedAccess}, but also achieves the topology's optimal coding gain of $\Lambda\lambda\gamma + 1$, albeit for the rather unrealistically demanding scenario\footnote{Indeed, thinking about either the wired server-and-workers setting or the dense cellular network scenario previously mentioned, it is more realistic for a user to be connected to very few cache-aided devices.} where $\lambda = (\Lambda - 1)/\Lambda\gamma$. For the same cyclic topology, the interesting work in~\cite{Reddy2020Rate-MemoryPlacement} designed a novel scheme for any $\lambda \geq 1$, which was --- for the similarly demanding regime of $\lambda \geq \Lambda/2$ and $\Lambda\gamma \leq 2$ --- proved to be at a factor of at most $2$ from the optimal under the assumption of uncoded placement.

Other notable works include the extension of the MACC model to support privacy and secrecy constraints~\cite{Namboodiri2021Multi-AccessPrivacy, Liang2021MultiaccessDemands, Namboodiri2021Multi-AccessDelivery}, the connection between MACC and structured index coding problems~\cite{Reddy2020StructuredCaching}, and the application of PDA designs to the multi-access setting~\cite{Sasi2021Multi-accessPDAs, Cheng2021ANetworks}. To the best of our understanding, it is fair to say that, under the assumption of a cyclic wrap-around topology and under realistic assumptions on the cache resources $\Lambda\gamma$ and on $\lambda$, the impact of the MACC setting has remained relatively modest.

Recently, a new MACC paradigm was presented in~\cite{Katyal2021Multi-AccessDesigns, Muralidhar2021ImprovedDesigns}, which involves previously unexplored powerful topologies that deviate from the cyclic topology originally proposed in~\cite{Hachem2017CodedAccess}. These two works in~\cite{Katyal2021Multi-AccessDesigns, Muralidhar2021ImprovedDesigns} drew a clever connection between coding for the MACC problem and employing a topology that is inspired by Cross Resolvable Designs (CRDs), where these CRDs constitute a special class of designs in combinatorics. The authors provided novel placement-and-delivery schemes from CRDs achieving a coding gain equal to $(q + 1)^z$, where $q$ and $z$ are two integer parameters such that each user is connected to $\lambda = qz$ distinct caches via a network topology that is implied by the chosen CRD. While though this gain nicely increases with $\lambda$, it does not increase with the cumulative cache redundancy $\Lambda\gamma$ and thus can remain relatively small as it does not capitalize on this redundancy.

A substantial breakthrough came with the very recent work in~\cite{Muralidhar2021Maddah-Ali-NiesenCaching}, which proposed a MACC model enjoying the same amount of resources $\lambda$ and $\Lambda\gamma$, but where now the users and the caches are connected following the well-known combinatorial topology of combination networks~\cite{Wan2018ACaches, Ji2015OnNetworks, ChiKinNgai2004NetworkNetworks, Xiao2007ANetworks}. This was a breakthrough because it allowed for the deployment of a subsequent scheme --- presented in~\cite{Muralidhar2021Maddah-Ali-NiesenCaching} as a generalization of the original Maddah-Ali and Niesen (MAN) scheme in~\cite{Maddah-Ali2014FundamentalCaching} --- that achieves an astounding coding gain $\binom{\Lambda\gamma + \lambda}{\lambda}$ far exceeding $\Lambda\gamma + 1$ even for small values of $\lambda$ and $\Lambda\gamma$, which is the regime that really matters.  A noticeable drawback of this new approach is that its performance is guaranteed only for a rather sparse range of $K \gg \Lambda$, and that it only captures the scenario where the users must all connect to an identical number of caches.

\subsection{Contributions}

Having as a starting point the combinatorial multi-access system model introduced in~\cite{Muralidhar2021Maddah-Ali-NiesenCaching}, we propose a model extension which allows for a denser range of possible number of users $K$ and for the coexistence of users that are connected to different numbers of caches. For this generalized combinatorial system model we extend the delivery scheme presented in~\cite{Muralidhar2021Maddah-Ali-NiesenCaching} to support the very large coding gains. We then proceed to prove this general scheme to be exactly optimal under the assumption of uncoded placement. We provide optimality by means of an information-theoretic converse that is based on index coding arguments. As a practical consequence of identifying the setting's exact fundamental limits, we now know that a basic and fixed MAN placement can optimally handle any generalized combinatorial network irrespective of having unknown numbers of users connected to different numbers of caches.

Subsequently, we consider a more general scenario that involves various ensembles of connectivities, including the ensemble of all possible connectivities as well as the smaller ensemble of those connectivities that simply abide by the constraint that each user is connected to the same number of $\lambda$ caches, without any additional structural constraint on the connectivity or on the number of users that each cache has to treat. For these settings, we develop information-theoretic converse bounds on the optimal average worst-case load, where the average is taken over the ensemble of interest, and where the optimal is over an optimized fixed placement. In particular, this converse on the average performance assumes optimal delivery for each connectivity and assumes an optimized uncoded placement that is fixed across all connectivities. This reflects a scenario where cache placement cannot be updated every time the topology changes\footnote{This setting may correspond to dense cellular networks where the connectivity may vary due to user mobility. This scenario may also correspond to a wired setting with a central server and multiple workers, where each worker is associated to any set of exactly $\lambda$ caches, but where this set is not known a priori due to a variety of reasons that may include a fluctuating network load or random node failures during data delivery.}.

\subsection{Paper Outline}

The rest of the paper is organized as follows. \Cref{sec: System Model} presents the system model together with some clarifying examples on the setting. Then, \Cref{sec: Main Results} presents the main results. Subsequently, \Cref{sec: Achievability Proof of MACC with Generalized Combinatorial Topology} presents the coding scheme for the generalized combinatorial topology, whereas~\Cref{sec: Converse Proof of MACC with Combinatorial Topology} presents the new matching converse bound. After this, \Cref{sec: Proof of Converse on Optimal Average Worst-Case Load} and~\Cref{sec: Proof of Converse on General Optimal Average Worst-Case Load} provide the proof of the converse on the optimal average worst-case load under uniformly random connectivity for the different considered ensembles, and finally \Cref{sec: Conclusions} concludes the paper. The appendices hold all additional proofs.

\subsection{Notation}

We denote by $\mathbb{Z}^{+}$ the set of positive integers and by $\mathbb{Z}^{+}_{0}$ the set of non-negative integers. For $n \in \mathbb{Z}^{+}$, we define $[n] \coloneqq \{1, \dots, n\}$. If $a, b \in \mathbb{Z}^{+}$ such that $a < b$, then $[a : b] \coloneqq \{a, a + 1, \dots, b - 1, b\}$. For $\alpha, \beta \in \mathbb{Z}^{+}$, we use $\alpha \mid \beta$ to denote that the integer $\alpha$ divides integer $\beta$. For sets we use calligraphic symbols, whereas for vectors we use bold symbols. Given a finite set $\mathcal{A}$, we denote by $|\mathcal{A}|$ its cardinality. Given $N$ sets $\{ \mathcal{S}_n : n \in [N] \}$, we denote by $\mathcal{S}^N = S_1 \times \dots \times S_N = \prod_{n \in [N]} \mathcal{S}_n = \{(s_1, \dots, s_N) : s_n \in \mathcal{S}_n, n \in [N]\}$ its $N$-ary Cartesian product. We use $\binom{n}{k}$ to denote the binomial coefficient $\frac{n!}{k!(n - k)!}$ and we let $\binom{n}{k} = 0$ whenever $n < 0$, $k < 0$ or $n < k$. We use the $\oplus$ symbol to denote the bitwise XOR operation.

\section{System Model}\label{sec: System Model}

We consider the centralized coded caching scenario where one single server has access to a library $\mathcal{L} = \{W_{n} : n \in [N]\}$ containing $N$ files of $B$ bits each. The server is connected to $K$ users through an error-free broadcast link. In the system there are $\Lambda$ caches, each of size $MB$ bits. Each user is associated (i.e., has full access) to a subset of these caches. We assume that the link between the server and the users is the main bottleneck, whereas we assume that the channel between each user and its assigned caches has infinite capacity. As is common, we assume that $N \geq K$. 

\subsection{Description of Connectivity}

Each of the $K$ users is connected to a subset of the $\Lambda$ caches. The way these connections are set up defines the network topology or connectivity. In the general case, different users are connected to a different number $\lambda \in [0 : \Lambda]$ of caches, where this value $\lambda$, depending on the user, can range from $\lambda = 0$ (corresponding to users that are not assisted by any cache) up to $\lambda = \Lambda$ (corresponding to the users that happen to be connected to all caches). What we will refer to as \emph{connectivity} will be here defined by the number of users that each $\lambda$-tuple $\mathcal{U}$ of caches is \emph{exactly} and \emph{uniquely} connected to, where again some of these sets $\mathcal{U}$ of caches can have size $\lambda = 1$ (sets consisting of a single cache), $\lambda = 2$ (where each set is a pair of caches), and so on. For instance, for the model in~\Cref{fig: MACC Model Example} the sets of $2$ caches $\{1, 2\}$, $\{2, 3\}$, $\{3, 4\}$ and $\{1, 4\}$ serve one user each. For the same model, we consider the set of $2$ caches $\{1, 3\}$, $\{2, 4\}$ to be serving no user, since there is no user connected exactly and uniquely to the caches in the set $\{1, 3\}$ (i.e., there is no user connected to only cache $1$ \emph{and} to cache $3$) as well as there is no user associated exactly and uniquely to the caches in the set $\{2, 4\}$ (i.e., there is no user connected to only cache $2$ \emph{and} to cache $4$). Similarly, we consider also the sets $\{1\}$, $\{2\}$, $\{3\}$ and $\{4\}$ to be serving no user, since there is no user connected exactly and uniquely to one cache only. On the other hand, for the model in~\Cref{fig: Connectivity Example} we can see that there are two users connected exactly and uniquely to cache $\{1\}$, there is one user associated exactly and uniquely to caches $\{1, 2, 3\}$ and there is one user connected exactly and uniquely to caches $\{2, 3, 4\}$. Hence, for such example the connectivity is defined by the number of users connected to the set of caches $\{1\}$, $\{1, 2, 3\}$ and $\{2, 3, 4\}$.

\begin{figure}[!htbp]
    \centering
    \tikzset{cnode/.style={draw, circle, minimum size = 1cm}}
    \tikzset{snode/.style={draw, rectangle, text centered}}
    \begin{tikzpicture}
        \node[snode, minimum height = 1cm, minimum width = 2cm](server){Server};
        \draw (server.south)--++(0, -1);
        \foreach \u/\n in {1/-3, 2/-1, 3/1, 4/3}{
            \draw ($(server.south) + (0, -1)$)--++(\n, -1) node[cnode, below](user\u){$\u$};
        }
        \node at (user1.west)[left, xshift = -0.25cm]{Users};
        \node at ($(server.south) + (3, -1)$)[right, xshift = 0.25cm]{Shared link};
        \foreach \c in {1, 2, 3, 4}{
            \node at ($(user\c.south) + (0, -1)$)[snode, below, minimum height = 1cm, minimum width = 1.5cm](cache\c){$\c$};
        }
        \draw (user1.south)--(cache1.north);
        \draw (user2.south)--(cache1.north);
        \draw (user3.south)--(cache1.north); \draw (user3.south)--(cache2.north); \draw (user3.south)--(cache3.north);
        \draw (user4.south)--(cache2.north); \draw (user4.south)--(cache3.north); \draw (user4.south)--(cache4.north);
        \node at (cache1.west)[left, xshift = -0.25cm]{Caches};
    \end{tikzpicture}
    \caption{Example of connectivity for the MACC model with $\Lambda = 4$.}
    \label{fig: Connectivity Example}
\end{figure}
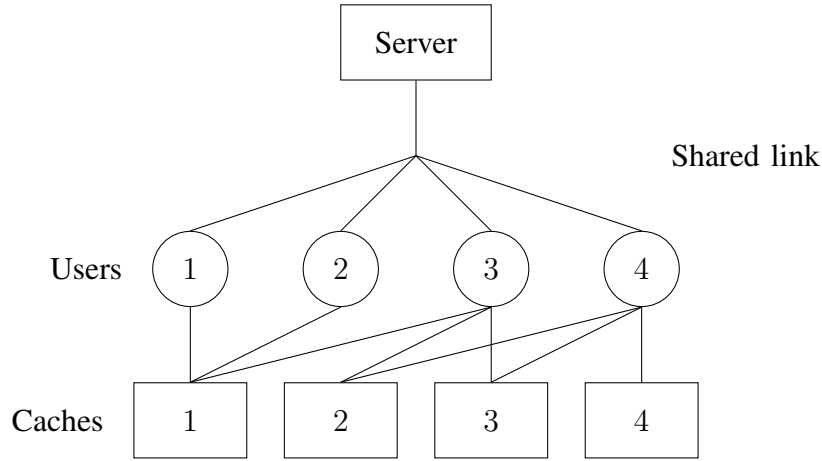

Let $\mathcal{B}$ be the set of all possible connectivites corresponding to the most general scenario, where each user can be arbitrarily connected to any set of caches without any constraint or structure. Each connectivity $b \in \mathcal{B}$ will be defined by the vector $\bm{K}_b = (K_{\mathcal{U}, b} : \mathcal{U} \subseteq [\Lambda])$, where we denote by $K_{\mathcal{U}, b} \in \mathbb{Z}^{+}_{0}$ the number of users associated exactly and uniquely to the caches in the set $\mathcal{U}$. Naturally, it holds that
\begin{equation}
    K = \sum_{\lambda \in [0 : \Lambda]} \sum_{\mathcal{U} \subseteq [\Lambda] : |\mathcal{U}| = \lambda} K_{\mathcal{U}, b}.
\end{equation}
For a given connectivity $b \in \mathcal{B}$, and for any set $\mathcal{U} \subseteq [\Lambda]$, we denote by $\mathcal{U}_{k}$ the $k$-th user connected\footnote{Note that this reflects the fact that connectivity is defined irrespective of any permutation of users. For example, consider $\Lambda = 2$ caches and $K = 3$ users. Then, if we have $K_{\{1\}, b} = 2$ and $K_{\{2\}, b} = 1$ for some connectivity $b \in \mathcal{B}$, it does not make any difference whether it is user~$1$ to be assigned to cache~$2$, and user~$2$ and user~$3$ to be assigned to cache~$1$; whether it is user~$2$ to be assigned to cache~$2$, and user~$1$ and user~$3$ to be assigned to cache~$1$; or whether it is user~$3$ to be assigned to cache~$2$, and user~$1$ and user~$2$ to be assigned to cache~$1$. All these scenarios correspond in fact to the same connectivity.} to the $|\mathcal{U}|$ caches in $\mathcal{U}$, where naturally $k \in [K_{\mathcal{U}, b}]$. If we let $\mathcal{K}_\lambda$ be the set of users which are each connected to exactly and uniquely $\lambda$ caches, it holds that
\begin{align}
    \mathcal{K} & = \bigcup_{\lambda \in [0 : \Lambda]} \{ \mathcal{K}_\lambda \} \\
    \mathcal{K}_\lambda & = \bigcup_{\mathcal{U} \subseteq [\Lambda] : |\mathcal{U}| = \lambda} \bigcup_{k \in [K_{\mathcal{U}, b}]}\{ \mathcal{U}_{k} \}
\end{align}
where $\mathcal{K}$ denotes the entire set of $K$ users in the system. This all holds for the most general setting corresponding to the ensemble $\mathcal{B}$, where this ensemble represents the set of all connectivities irrespective of how different these connectivities are from one another. We will revisit this general scenario later, when we will calculate the converse bound on the optimal average performance over all connectivities in $\mathcal{B}$.

A more restricted class of connectivities can be found in the ensemble $\mathcal{B}_\lambda$, which consists of all those connectivities for which each user is connected to exactly $\lambda$ caches for some fixed number $\lambda \in [0 : \Lambda]$. Any connectivity that satisfies this constraint belongs to $\mathcal{B}_\lambda$. For example, a connectivity $b$ belongs to $\mathcal{B}_2$ if and only if this connectivity guarantees that every user is connected to exactly $\lambda = 2$ caches. The well-known cyclic wrap-around connectivity depicted in Figure~\ref{fig: MACC Model Example} is one of the connectivities in this class $\mathcal{B}_2$, since it guarantees that each user is connected to exactly $\lambda = 2$ caches. There exist many additional connectivities that belong to $\mathcal{B}_2$. We will revisit this class $\mathcal{B}_\lambda$ when we will calculate the average optimal performance, averaged across all its connectivities. 

A broader class of connectivities is simply the $(\Lambda + 1)$-ary Cartesian product
\begin{equation}
    \mathcal{B}^{\Lambda + 1} = \prod_{\lambda = 0}^{\Lambda} \mathcal{B}_\lambda = \{(b_0, \dots, b_\Lambda) : b_\lambda \in \mathcal{B}_\lambda, \lambda \in [0 : \Lambda]\}
\end{equation}
which consists of all connectivities that guarantee that some $K'_\lambda$ users --- from the $K = \sum_{\lambda \in [0 : \Lambda]} K'_\lambda$ users in total --- are each connected to exactly $\lambda$ caches for each $\lambda \in [0 : \Lambda]$ and for a fixed set of integers $K'_0, K'_1, \dots, K'_\Lambda$. For example, a connectivity $b$ belongs to $\mathcal{B}_2 \times \mathcal{B}_3$ if and only if this connectivity guarantees that $K'_2$ users are connected to $2$ caches, and the rest $K'_3 = K - K'_2$ users are connected to $3$ caches, for any fixed pair $K'_2, K'_3$ such that $K'_2 + K'_3 = K$. For us this class is important because it captures the generalization of the aforementioned combinatorial\footnote{As already specified in~\Cref{sec: MACC Introduction}, the topology introduced in~\cite{Muralidhar2021Maddah-Ali-NiesenCaching} follows the well-known combinatorial nature of combination networks. To avoid any confusion with the well-defined term \emph{combination-network topology} which is prevalent in the literature of network coding, we will use the simplified term \emph{combinatorial topology} to refer to the topology in~\cite{Muralidhar2021Maddah-Ali-NiesenCaching}.} topology, for which we will calculate the optimal performance under the assumption of uncoded placement.

To avoid heavy notation, the dependence on the connectivity $b$ will always be suppressed and left implied. For example, while $\bm{K}_b$ is a function of the specific connectivity, this dependence will be implied when we henceforth use the notation $\bm{K}$ instead of $\bm{K}_b$.  An exception to this rule will be allowed when considering the number of users $K_{\mathcal{U}, b}$ associated to the caches in $\mathcal{U}$ for some $\mathcal{U} \subseteq [\Lambda]$. 

In terms of file-requests, we use the notation $W_{d_{\mathcal{U}_{k}}}$ to denote the file requested by the user identified by $\mathcal{U}_{k}$, which we remind the reader is simply the $k$-th user connected exactly and uniquely to the caches in the set $\mathcal{U}$. For the sake of simplicity, we denote by $\bm{d} = (\bm{d}_0, \dots, \bm{d}_\Lambda)$ the demand vector containing the indices of the files requested by the users in the system, where
\begin{equation}
    \bm{d}_{\lambda} \coloneqq \left(\bm{d}_{\mathcal{U}, [K_{\mathcal{U}, b}]} : \mathcal{U} \subseteq [\Lambda], |\mathcal{U}| = \lambda \right)
\end{equation}
represents the vector containing the indices of the files requested by all the users connected to exactly $\lambda$ caches and where $\bm{d}_{\mathcal{U}, [K_{\mathcal{U}, b}]} \coloneqq (d_{\mathcal{U}_{1}}, \dots, d_{\mathcal{U}_{K_{\mathcal{U}, b}}})$ represents the vector containing the indices of the files requested by the $K_{\mathcal{U}, b}$ users connected to the caches in the set $\mathcal{U}$ for a given connectivity $b \in \mathcal{B}$. To account for the possibility that this notation is hard to follow, we offer clarifying examples later on.

\subsection{Generalized Combinatorial Topology}\label{sec: Generalized Combinatorial Topology}

Directly from the aforementioned Cartesian product class, there is a particular connectivity (topology) $ b\in\mathcal{B}^{\Lambda + 1}$ that is of special interest to us. This topology, which we refer to as the \emph{generalized combinatorial topology}, guarantees that any one set of $\lambda$ caches is uniquely assigned to $K_\lambda$ users, and this holds for every $\lambda \in [0 : \Lambda]$. Given the nature of the connectivity, we have that\footnote{As one would expect, we assume for such combinatorial connectivity that $\binom{\Lambda}{\lambda} \mid K'_\lambda$ for each $\lambda \in [0 : \Lambda]$.} $K_{\mathcal{U}, b} = K_\lambda = K'_\lambda/\binom{\Lambda}{\lambda}$ for each $\mathcal{U} \subseteq [\Lambda]$ with $|\mathcal{U}| = \lambda$ and $\lambda \in [0 : \Lambda]$. We also have that
\begin{equation}
    K = \sum_{\lambda = 0}^{\Lambda}K_\lambda\binom{\Lambda}{\lambda}.
\end{equation}
To clarify, the term $K'_\lambda$ again describes the total number of users each of which is associated to exactly $\lambda$ caches, while the term $K_\lambda$ is the normalization of $K'_\lambda$ and it describes the total number of users uniquely served by any one set of $\lambda$ caches. For this generalized combinatorial topology, we collect all the $K_\lambda$ terms in the vector $\bm{K}_{\text{comb}} = (K_0, \dots, K_\Lambda)$.

\begin{example}[$\Lambda = 4, \bm{K}_{\text{comb}} = (0, 0, 1, 0, 0)$]
    Consider the MACC problem with the original combinatorial topology in~\cite{Muralidhar2021Maddah-Ali-NiesenCaching} and $\Lambda = 4$ caches. Since $\bm{K}_{\text{comb}} = (0, 0, 1, 0, 0)$, each set of $\lambda = 2$ caches is uniquely assigned to $K_2 = 1$ user, hence there are $K = \sum_{\lambda = 0}^{\Lambda}K_\lambda\binom{\Lambda}{\lambda}  = \binom{4}{2} = 6$ users in total. Recalling that each user is identified by the set of the $\lambda = 2$ caches it is connected to as well as by its index $k \in [K_2]$, we write the set of users $\mathcal{K}$ as
    \begin{align}
        \mathcal{K} & = \bigcup_{\lambda \in [0 : \Lambda]} \bigcup_{\mathcal{U} \subseteq [\Lambda] : |\mathcal{U}| = \lambda} \bigcup_{k \in [K_\lambda]}\{\mathcal{U}_{k}\} \\
                    & = \bigcup_{\mathcal{U} \subseteq [\Lambda] : |\mathcal{U}| = 2} \bigcup_{k \in [K_2]}\{\mathcal{U}_{k}\} \\
                    & = \left\{ \{1, 2\}_1, \{1, 3\}_1, \{1, 4\}_1, \{2, 3\}_1, \{2, 4\}_1, \{3, 4\}_1 \right\}.
    \end{align}
    Notice that for ease of notation we will often omit braces and commas when indicating sets, hence we can also write $\mathcal{K} = \left\{ {12}_1, {13}_1, {14}_1, {23}_1, {24}_1, {34}_1 \right\}$. These are the $6$ users in the system. The demand vector is given by $\bm{d} = (\bm{d}_0, \bm{d}_1, \bm{d}_2, \bm{d}_3, \bm{d}_4) = (0, \dots, 0, \bm{d}_2, 0, \dots, 0)$, where
    \begin{align}
        \bm{d}_2 & = \left(\bm{d}_{\mathcal{U}, [K_2]} : \mathcal{U} \subseteq [\Lambda], |\mathcal{U}| = 2 \right) \\
                 & = \left(d_{{12}_1}, d_{{13}_1}, d_{{14}_1}, d_{{23}_1}, d_{{24}_1}, d_{{34}_1}\right)
    \end{align}
    recalling that $\bm{d}_{\mathcal{U}, [K_2]} = (d_{\mathcal{U}_1}, \dots, d_{\mathcal{U}_{K_2}})$. The setting described in this example corresponds to the model introduced in~\cite{Muralidhar2021Maddah-Ali-NiesenCaching}, as it will be pointed out also later.
\end{example}

\begin{example}[$\Lambda = 4, \bm{K}_{\text{comb}} = (0, 0, 2, 2, 0)$]
    Let us consider now the following more involved MACC problem with a generalized combinatorial topology and again $\Lambda = 4$ caches. Since $\bm{K}_{\text{comb}} = (0, 0, 2, 2, 0)$, each set of $2$ caches is uniquely connected to $K_2 = 2$ users and each set of $3$ caches is uniquely connected to $K_3 = 2$ users, which tells us that there are $K = \sum_{\lambda = 0}^{\Lambda} K_\lambda \binom{\Lambda}{\lambda} = 2\binom{4}{2} + 2\binom{4}{3} = 20$ users in total. The set of users $\mathcal{K}$ is given by
    \begin{align}
        \mathcal{K} & = \bigcup_{\lambda \in [0 : \Lambda]} \bigcup_{\mathcal{U} \subseteq [\Lambda] : |\mathcal{U}| = \lambda} \bigcup_{k \in [K_\lambda]}\{\mathcal{U}_{k}\} \\
                    & = \bigcup_{\lambda \in [2 : 3]} \bigcup_{\mathcal{U} \subseteq [\Lambda] : |\mathcal{U}| = \lambda} \bigcup_{k \in [K_\lambda]}\{\mathcal{U}_{k}\} \\
                    & = \{ {12}_1, {12}_2, {13}_1, {13}_2, {14}_1, {14}_2, {23}_1, {23}_2, {24}_1, {24}_2, {34}_1, {34}_2,  \\
                    & \phantom{{} = {} \{ } {123}_1, {123}_2, {124}_1, {124}_2, {134}_1, {134}_2, {234}_1, {234}_2 \}.
    \end{align}
    As a small reminder, users ${12}_1,{12}_2$ are the first and second users connected to the pair of caches $\{1, 2\}$, while user ${234}_2$ is the second of two users connected to the caches in the triplet $\{2, 3, 4\}$ and to no other cache. In this case, the demand vector is given by $\bm{d} = (\bm{d}_0, \bm{d}_1, \bm{d}_2, \bm{d}_3, \bm{d}_4) = (0, \dots, 0, \bm{d}_2, \bm{d}_3, 0, \dots, 0)$, where
    \begin{align}
        \bm{d}_2 & = \left(\bm{d}_{\mathcal{U}, [K_2]} : \mathcal{U} \subseteq [\Lambda], |\mathcal{U}| = 2 \right) \\
                 & = \left(d_{{12}_1}, d_{{12}_2}, d_{{13}_1}, d_{{13}_2}, d_{{14}_1}, d_{{14}_2}, d_{{23}_1}, d_{{23}_2}, d_{{24}_1}, d_{{24}_2}, d_{{34}_1}, d_{{34}_2}\right) \\
        \bm{d}_3 & = \left(\bm{d}_{\mathcal{U}, [K_3]} : \mathcal{U} \subseteq [\Lambda], |\mathcal{U}| = 3 \right) \\
                 & = \left(d_{{123}_1}, d_{{123}_2}, d_{{124}_1}, d_{{124}_2}, d_{{134}_1}, d_{{134}_2}, d_{{234}_1}, d_{{234}_2}\right)
    \end{align}
    recalling that $\bm{d}_{\mathcal{U}, [K_2]} = (d_{\mathcal{U}_1}, \dots, d_{\mathcal{U}_{K_2}})$ and $\bm{d}_{\mathcal{U}, [K_3]} = (d_{\mathcal{U}_1}, \dots, d_{\mathcal{U}_{K_3}})$.
\end{example}

\subsection{Worst-Case Load and Average Worst-Case Load}

As in the original coded caching scenario, the communication procedure is split into the \emph{placement phase} and the \emph{delivery phase}. During the placement phase --- which typically occurs well before the delivery phase --- the central server fills the caches without any knowledge of the future requests from the users. The delivery phase --- which typically happens when the network is saturated and interference-limited --- commences when the file-requests of all users are simultaneously revealed. During delivery, the server sends coded messages over the bottleneck shared link, so that each user can retrieve the missing information from the broadcast transmission. The users cancel the interference terms that appear in the broadcast transmission, and do so by means of the cached contents they have access to, eventually decoding their own messages.

The aforementioned placement phase may or may not be aware of the given topology $b \in \mathcal{B}$ that will be encountered during delivery. Both these scenarios of \emph{topology-aware} and \emph{topology-agnostic} cache placement will be addressed.

In the topology-aware scenario, given a unique topology $b \in \mathcal{B}$ that is known throughout placement and delivery, the worst-case communication load $R_{b}$ is defined as the total number of transmitted bits, normalized by the file size $B$, that can guarantee the correct delivery of any $K$-tuple of requested files in the worst-case scenario. The optimal communication load $R^\star_{b}$ is consequently defined as
\begin{equation}
    R^\star_{b}(M) \coloneqq \inf\{R_{b} : \text{ $(M, R_{b})$ is achievable}\}
\end{equation}
where the tuple $(M, R_{b})$ is said to be \emph{achievable} if there exists a caching-and-delivery procedure for which, for any possible demand, a load $R_{b}$ can be guaranteed. This metric captures the optimal performance, optimized over all topology-aware placement-and-delivery schemes. This is indeed the metric that is common in the majority of works on coded caching. For the particular case of the generalized combinatorial topology, this worst-case communication load will be denoted by $R^{\star}_{\text{comb}}$.

On the other hand, to capture the fact that topologies may vary in time often much faster than the rate with which caches can be updated, we will also consider the scenario where a fixed placement must be designed to handle an ensemble of possible topologies\footnote{This means that the connectivity is not known a priori and the cache placement cannot be modified whenever a new connectivity is presented. }. In this topology-agnostic scenario, we will employ an average metric that captures the average performance of coded caching over the ensemble of topologies. In particular, we will consider the average worst-case communication load $R_{\text{avg}} = \mathbb{E}_b[R_b]$, which is defined as the expected number of transmitted bits (averaged over a specified ensemble of connectivities, and normalized by the file size $B$) that can guarantee the correct delivery of all requested files irrespective of the request. 
When the averaging is done over the entire connectivity ensemble $\mathcal{B}$ of equiprobable connectivities, the corresponding optimal average worst-case load $R_{\text{avg}, \mathcal{B}}^\star$ is defined and denoted as
\begin{equation}
    R_{\text{avg}, \mathcal{B}}^\star(M) \coloneqq \inf\{R_{\text{avg}, \mathcal{B}} : \text{ $(M, R_{\text{avg}, \mathcal{B}})$ is achievable}\}
\end{equation}
where the tuple $(M, R_{\text{avg}, \mathcal{B}})$ is said to be \emph{achievable} if there exists a joint placement-and-delivery method with an optimized fixed placement phase\footnote{Here, the optimal performance is over the class of all schemes that employ a cache placement which can be chosen and optimized, but which must remain fixed for all connectivities in the ensemble. The scheme is free to employ delivery methods that are fully aware of the current topology and can adapt to it. For every choice of fixed placement, and then for every connectivity, there is a minimum amount of bits to be sent. We are interested in minimizing the average of these amounts of bits to be transmitted, averaged over all the connectivities in the ensemble of focus.} for which an average load $R_{\text{avg}, \mathcal{B}}$ can be guaranteed, where the averaging is over the connectivity ensemble of focus. Similarly, when the averaging is done over the smaller symmetric ensemble $\mathcal{B}_\lambda$, the optimal average performance will be denoted by $R^{\star}_{\text{avg}, \mathcal{B}_\lambda}$.

\section{Main Results}\label{sec: Main Results}

We present in this section the main results. Firstly, we identify the fundamental limits of the system model described in~\Cref{sec: Generalized Combinatorial Topology} corresponding to the unique generalized combinatorial topology. This will identify the optimal performance of the generalized combinatorial connectivity, optimized over all coded caching schemes under the assumption of uncoded prefetching. We will then proceed to study the performance over ensembles of connectivities. Taking into account the scenario where there are $K = K'_\lambda$ users and each of them is connected to a set of exactly $\lambda$ caches, where such set is picked uniformly at random among all possible sets of $\lambda$ caches, we develop a converse bound on the optimal average worst-case load assuming the connectivities in $\mathcal{B}_\lambda$ to be equiprobable for a fixed $\lambda \in [\Lambda]$. Finally, such bound is further extended to consider the most general ensemble $\mathcal{B}$ of all possible connectivities.

\subsection{Multi-Access Coded Caching with Generalized Combinatorial Topology}

Our first result is obtained by extending the achievable scheme proposed in~\cite{Muralidhar2021Maddah-Ali-NiesenCaching} and by developing a matching converse bound based on the well-known acyclic subgraph index coding method. The result is formally stated in the following theorem.

\begin{theorem}\label{thm: MACC with Generalized Combinatorial Topology}
    Consider the multi-access coded caching problem with $\Lambda$ caches and the generalized combinatorial topology described in~\Cref{sec: Generalized Combinatorial Topology}. Under the assumption of uncoded cache placement, the optimal worst-case communication load $R^\star_{\text{comb}}$ is a piece-wise linear curve with corner points
    \begin{equation}
        (M, R^\star_{\text{comb}}) = \left(t \frac{N}{\Lambda}, \sum_{\lambda = 0}^{\Lambda - t}K_\lambda \frac{\binom{\Lambda}{t + \lambda}}{\binom{\Lambda}{t}} \right), \quad \forall t \in [0 : \Lambda].
    \end{equation}
\end{theorem}

\begin{proof}
    The coded caching scheme is described in~\Cref{sec: Achievability Proof of MACC with Generalized Combinatorial Topology}, whereas the information-theoretic converse is presented in~\Cref{sec: Converse Proof of MACC with Combinatorial Topology}.
\end{proof}

Directly from the converse in~\Cref{thm: MACC with Generalized Combinatorial Topology}, and from the application as in~\Cref{sec: Achievability Proof of MACC with Generalized Combinatorial Topology} of the scheme in~\cite{Muralidhar2021Maddah-Ali-NiesenCaching}, we now have the following corollary. To place the corollary in context, we note that typically (see for example~\cite{Reddy2020Rate-MemoryPlacement}) cache placements are specifically calibrated to reflect the cache-connectivity capability of each user.

\begin{corollary}\label{cor: Universality of MAN Placement}
    The basic $\Lambda$-cache MAN placement allows for the optimal performance for any instance of the generalized combinatorial topology. This means that, as long as the connectivity is from the generalized combinatorial topology, then the single MAN placement yields a uniformly optimal performance irrespective of the cache-connectivity capability of each user, i.e., irrespective of how many users are connected to how many caches. 
\end{corollary}

\subsection{Topology Ensembles Analysis and Converse Bound Development on the Optimal Average Worst-Case Load}

It is conceivable to have a scenario where each user is connected to the same number of caches $\lambda \in [\Lambda]$, but where the actual connectivity is not known in advance. Our second contribution is the development of a converse bound on the optimal average worst-case load for a fixed value\footnote{Clearly, the scenario $\lambda = 0$ is trivial. Indeed, if all users $K$ are connected to no cache, there is only one connectivity which is optimally served with load equal to $K$.} of $\lambda \in [\Lambda]$, corresponding to the ensemble $\mathcal{B}_\lambda$, where each connectivity in $\mathcal{B}_\lambda$ is assumed to be equiprobable. The result is stated in the following theorem.

\begin{theorem}\label{thm: Converse on Optimal Average Worst-Case Load}
    Consider the ensemble $\mathcal{B}_\lambda$ of multi-access coded caching problems with $\Lambda$ caches and $K = K'_\lambda$ users each connected to exactly $\lambda \in [\Lambda]$ caches. Under the assumption of fixed uncoded cache placement and equiprobable connectivities, the optimal average worst-case communication load $R^\star_{\text{avg}, \mathcal{B}_\lambda}$ is lower bounded by $R_{\text{avg}, \mathcal{B}_\lambda, \text{LB}}$ which is a piece-wise linear curve with corner points
    \begin{equation}
        (M, R_{\text{avg}, \mathcal{B}_\lambda, \text{LB}}) = \left(t \frac{N}{\Lambda}, \frac{K'_\lambda \binom{\Lambda}{t + \lambda}}{\binom{\Lambda}{\lambda}\binom{\Lambda}{t}} + A_t \right), \quad \forall t \in [0 : \Lambda - \lambda + 1]
    \end{equation}
    where we define
    \begin{equation}
        A_t \coloneqq \frac{K'_\lambda}{\left| \mathcal{B}_\lambda \right|}\binom{\Lambda - t}{\lambda}\left(1 - \frac{1}{\binom{t + \lambda}{\lambda}} \right).
    \end{equation}
\end{theorem}

\begin{proof}
  The proof is reported in~\Cref{sec: Proof of Converse on Optimal Average Worst-Case Load}.
\end{proof}

In the following we offer an interesting comparison between the results in~\Cref{thm: MACC with Generalized Combinatorial Topology} and in~\Cref{thm: Converse on Optimal Average Worst-Case Load}, after setting $K = K'_\lambda = K_\lambda \binom{\Lambda}{\lambda}$ to be the same in both cases.

\begin{corollary}\label{cor: Combinatorial Topology and Converse Comparison}
    For a fixed unique $\lambda$ and a fixed $K = K'_\lambda$ such that $\binom{\Lambda}{\lambda} \mid K'_\lambda$, then
    \begin{equation}
        R_{\text{avg}, \mathcal{B}_\lambda, \text{LB}} > R^\star_{\text{comb}}, \quad \forall t \in [\Lambda - \lambda]
    \end{equation}
which simply says that, for non-trivial values of $t$, the optimal average worst-case communication load in~\Cref{thm: Converse on Optimal Average Worst-Case Load} is strictly larger than the optimal worst-case communication load of the combinatorial connectivity in~\Cref{thm: MACC with Generalized Combinatorial Topology}.
\end{corollary}
\begin{proof}
    The proof follows immediately after observing that $A_t > 0$ for each $t \in [\Lambda - \lambda]$.
\end{proof}

We now transition to the most general scenario where any possible connectivity is allowed, namely, where each user can be arbitrarily connected to any set of caches without any constraint or structure. This corresponds to the ensemble $\mathcal{B}$. Our third contribution is the development of a converse bound on the optimal average worst-case load assuming the connectivities in the set $\mathcal{B}$ to be equiprobable.

\begin{theorem}\label{thm: Converse on General Optimal Average Worst-Case Load}
    Consider the ensemble $\mathcal{B}$ of multi-access coded caching problems with $\Lambda$ caches and $K$ users, where each of them can arbitrarily connect to any set of caches without any constraints. Under the assumption of fixed uncoded cache placement and equiprobable connectivities, the optimal average worst-case communication load $R^\star_{\text{avg}, \mathcal{B}}$ is lower bounded by $R_{\text{avg}, \mathcal{B}, \text{LB}}$ which is a piece-wise linear curve with corner points
    \begin{equation}
        (M, R_{\text{avg}, \mathcal{B}, \text{LB}}) = \left(t \frac{N}{\Lambda}, \sum_{\lambda = 0}^{\Lambda}\frac{K \binom{\Lambda}{t + \lambda}}{2^{\Lambda}\binom{\Lambda}{t}} + A_{t, \lambda} \right), \quad \forall t \in [0 : \Lambda]
    \end{equation}
    where we define
    \begin{equation}
        A_{t, \lambda} \coloneqq \frac{K}{\left| \mathcal{B} \right|}\binom{\Lambda - t}{\lambda}\left(1 - \frac{1}{\binom{t + \lambda}{\lambda}} \right).
    \end{equation}
\end{theorem}

\begin{proof}
  The proof is reported in~\Cref{sec: Proof of Converse on General Optimal Average Worst-Case Load}.
\end{proof}

\section{Achievability Proof of \texorpdfstring{\Cref{thm: MACC with Generalized Combinatorial Topology}}{Theorem~\ref{thm: MACC with Generalized Combinatorial Topology}}}\label{sec: Achievability Proof of MACC with Generalized Combinatorial Topology}

We devote this section to presenting the general placement-and-delivery scheme, which will allow us to prove that the load performance in~\Cref{thm: MACC with Generalized Combinatorial Topology} is indeed achievable. Recall that this is for the case of the generalized combinatorial topology presented in~\Cref{sec: Generalized Combinatorial Topology}. As a quick reminder, under this topology, out of the total of $K = \sum_{\lambda = 0}^\Lambda K'_\lambda$ users, there exist $K'_\lambda$ users each of which is associated to exactly $\lambda$ caches. Similarly, the normalized $K_\lambda$ simply describes the total number of users uniquely served by any one set of $\lambda$ caches. As it will be clear in a short while, we point out that the general delivery scheme here proposed is a properly calibrated orthogonal concatenation of the transmitted sequences in~\cite{Muralidhar2021Maddah-Ali-NiesenCaching} for different values of $\lambda \in [0 : \Lambda]$.

\subsection{Description of the General Scheme}\label{sec: Description of the General Scheme}

The communication process is split into the placement phase and the delivery phase. Both phases are designed with full knowledge\footnote{Indeed, there is generally some flexibility and the placement can be designed for a specific topology, when such topology is known. In fact, for the generalized combinatorial topology the basic $\Lambda$-cache MAN placement is enough to achieve the optimal performance, as mentioned in~\Cref{cor: Universality of MAN Placement}.} of the topology, i.e., with knowledge of the fact that during delivery the users are connected to the caches according to the unique generalized combinatorial topology in~\Cref{sec: Generalized Combinatorial Topology}. 

\subsubsection{Placement Phase}\label{sec: Placement Phase}

This procedure is performed by the central server without knowing the future requests of each user. Let $M \coloneqq t N/\Lambda$ be the volume of data, in units of file, that each of the $\Lambda$ caches can store, where $t \in [0 : \Lambda]$ is an integer value. Each file is split into $\binom{\Lambda}{t}$ equal-sized non-overlapping subfiles as follows
\begin{equation}
  W_{n} = \{W_{n, \mathcal{T}} : \mathcal{T} \subseteq [\Lambda], |\mathcal{T}| = t\}, \quad \forall n \in [N]
\end{equation}
and the memory of the $\ell$-th cache is filled as 
\begin{equation}
  \mathcal{Z}_{\ell} = \{W_{n, \mathcal{T}} : n \in [N], \mathcal{T} \subseteq [\Lambda], |\mathcal{T}| = t, \ell \in \mathcal{T}\}
\end{equation}
for each $\ell \in [\Lambda]$. A quick calculation allows us to verify that each cache $\ell \in [\Lambda]$ stores a total of
\begin{equation}
    |\mathcal{Z}_{\ell}| = N\binom{\Lambda - 1}{t - 1}\frac{B}{\binom{\Lambda}{t}} = t\frac{N}{\Lambda}B = MB
\end{equation}
bits, which guarantees that the memory-size constraint is satisfied.

\subsubsection{Delivery Phase}

The delivery phase takes place once the file-requests of the users are revealed. Let $\bm{d} = (\bm{d}_0, \dots, \bm{d}_\Lambda)$ be the demand vector containing the indices of the files demanded by the $K$ users in $\mathcal{K}$. Then, the server transmits $X_{\bm{d}} = (X_0, \dots, X_{\Lambda - t})$, where $X_\lambda = \left( X_{\lambda, 1}, \dots, X_{\lambda, K_\lambda}\right)$ and where
\begin{equation}\label{eqn: Sequence of Multicast Messages}
    X_{\lambda, k} = \left( \bigoplus_{\mathcal{U} \subseteq \mathcal{S} : |\mathcal{U}| = \lambda } W_{d_{\mathcal{U}_{k}}, \mathcal{S} \setminus \mathcal{U}} : \mathcal{S} \subseteq [\Lambda], |\mathcal{S}| = t + \lambda \right)
\end{equation}
for each $k \in [K_\lambda]$ and for each $\lambda \in [0: \Lambda]$. Simply, $X_{\lambda, k}$ corresponds to the broadcast transmission which successfully delivers all the missing information to all the $k$-th users connected to any $\lambda$ caches. 
Notice that users connected to more than $\Lambda - t$ caches have access to the entire library, hence no transmission is needed for them.

\begin{remark}
    We point out that for any one fixed value of $\lambda \in [0 : \Lambda]$ and $k \in [K_\lambda]$, the sequence of multicast messages in~\eqref{eqn: Sequence of Multicast Messages} matches the scheme in~\cite{Muralidhar2021Maddah-Ali-NiesenCaching}. However, whereas the work in~\cite{Muralidhar2021Maddah-Ali-NiesenCaching} considers $K_\lambda = 1$ and only one single value of $\lambda \in [0 : \Lambda]$ at a time (i.e., considers that every user is connected to exactly $\lambda$ caches), the scheme here considers any $K_\lambda \in \mathbb{Z}^{+}_{0}$ and most importantly considers all possible values of $\lambda \in [0 : \Lambda]$ at the same time, thus capturing a rather general scenario where the set of $K$ users involves users with unequal cache-connectivity capabilities (i.e., users that have access to an unequal number of caches). The most interesting outcome of our generalization is the rather surprising fact that a basic TDMA-like approach of treating groups of users with different cache-connectivity capabilities is in fact optimal, as we will see soon. This is indeed surprising, because users with different cache-connectivity capabilities still maintain an abundance of common side information, which could have conceivably been exploited using joint encoding across the groups. The optimality of our scheme reveals that there is no need to encode across users with different cache-connectivity capabilities, and that a TDMA-like approach is optimal, even though there is abundant additional opportunities to encode across different groups of users that are now treated separately. 
\end{remark}

\begin{remark}
    An additional interesting consideration regards the adopted cache placement. To the best of our knowledge, the cache placement in the multi-access scenario usually changes substantially as $\lambda$ changes, e.g., this is what happens for the cyclic wrap-around topology (see~\cite{Reddy2020Rate-MemoryPlacement}). However, as described in \Cref{cor: Universality of MAN Placement}, for the combinatorial topology in~\cite{Muralidhar2021Maddah-Ali-NiesenCaching} the same MAN cache placement holds for any value of $\lambda \in [0 : \Lambda]$. And such MAN placement not only works for any distinct value of $\lambda$, but even allows to handle \emph{optimally} all the $\lambda$-instances when all such instances appear simultaneously at the same time. In other words, the MAN cache placement --- which is independent of the number of caches each user is connected to, and is a function of the number of caches $\Lambda$ and of the cache redundancy $\Lambda\gamma$ only --- can be considered a single unified cache placement approach for which the least possible worst-case communication load is achieved under uncoded cache placement, even when each of the $\lambda$-settings in~\cite{Muralidhar2021Maddah-Ali-NiesenCaching} is taken into account simultaneously.
\end{remark}

Even though the proof of correctness of the delivery in~\eqref{eqn: Sequence of Multicast Messages}, for a specific $\lambda$, was reported in~\cite{Muralidhar2021Maddah-Ali-NiesenCaching}, we briefly describe for completeness how decoding is achieved in~\eqref{eqn: Sequence of Multicast Messages}. Consider a user $\mathcal{U}'_k$ for some $\mathcal{U}' \subseteq [\Lambda]$ with $|\mathcal{U}'| = \lambda$, for some $k \in [K_\lambda]$ and for some $\lambda \in [0 : \Lambda]$. Consider a specific $\mathcal{S} \subseteq [\Lambda]$ for which $|\mathcal{S}| = t + \lambda$ and $\mathcal{U}' \subseteq \mathcal{S}$. For such set $\mathcal{S}$, the coded transmission
\begin{equation}
    \bigoplus_{\mathcal{U} \subseteq \mathcal{S} : |\mathcal{U}| = \lambda } W_{d_{\mathcal{U}_{k}}, \mathcal{S} \setminus \mathcal{U}}
\end{equation}
is sent. Since $\mathcal{U}' \subseteq \mathcal{S}$, we can rewrite the coded transmission as
\begin{equation}
    \bigoplus_{\mathcal{U} \subseteq \mathcal{S} : |\mathcal{U}| = \lambda } W_{d_{\mathcal{U}_{k}}, \mathcal{S} \setminus \mathcal{U}} = W_{d_{\mathcal{U}'_{k}}, \mathcal{S} \setminus \mathcal{U}'} \oplus  \underbrace{\bigoplus_{\mathcal{U} \subseteq \mathcal{S} : |\mathcal{U}| = \lambda, \mathcal{U} \neq \mathcal{U}' } W_{d_{\mathcal{U}_{k}}, \mathcal{S} \setminus \mathcal{U}}}_{\text{interference}}.
\end{equation}
Notice that user $\mathcal{U}'_k$ can correctly decode the subfile $W_{d_{\mathcal{U}'_{k}}, \mathcal{S} \setminus \mathcal{U}'}$. Indeed, this user has access to all subfiles in the interference term, considering that $\mathcal{S} \setminus \mathcal{U} \cap \mathcal{U}' \neq \emptyset$ since $\mathcal{U} \neq \mathcal{U}'$. User $\mathcal{U}'_k$ can consequently decode a distinct subfile for each $\mathcal{S} \subseteq [\Lambda]$ with $|\mathcal{S}| = t + \lambda$ and $\mathcal{U}' \subseteq \mathcal{S}$. Since there is a total of $\binom{\Lambda - \lambda}{t}$ such sets $\mathcal{S}$ and user $\mathcal{U}'_k$ misses a total of $\binom{\Lambda - \lambda}{t}$ subfiles, we can conclude that user $\mathcal{U}'_k$ correctly decodes all missing subfiles from the coded transmission in~\eqref{eqn: Sequence of Multicast Messages}. Clearly, the same holds also for any other user, hence showing the decodability of the scheme.

\subsection{Performance Calculation}

To evaluate the performance of the scheme proposed in~\Cref{sec: Description of the General Scheme}, it is enough to calculate $|X_{\bm{d}}|/B$. Since it can be easily checked that
\begin{equation}
    \frac{|X_{\lambda, k}|}{B} = \frac{\binom{\Lambda}{t + \lambda}}{\binom{\Lambda}{t}}
\end{equation}
the achievable load performance is equal to
\begin{align}
    R_{\text{comb}, \text{UB}} & = \frac{|X_{\bm{d}}|}{B} \\
                               & = \sum_{\lambda = 0}^{\Lambda - t}\frac{|X_\lambda|}{B} \\
                               & = \sum_{\lambda = 0}^{\Lambda - t}\sum_{k = 1}^{K_\lambda}\frac{|X_{\lambda, k}|}{B} \\
                               & = \sum_{\lambda = 0}^{\Lambda - t}K_\lambda \frac{\binom{\Lambda}{t + \lambda}}{\binom{\Lambda}{t}}
\end{align}
and thus it holds that $R^\star_{\text{comb}} \leq R_{\text{comb}, \text{UB}}$, where this upper bound is a piece-wise linear curve with corner points
\begin{equation}
  (M, R_{\text{comb}, \text{UB}}) = \left(t\frac{N}{\Lambda}, \sum_{\lambda = 0}^{\Lambda - t}K_\lambda \frac{\binom{\Lambda}{t + \lambda}}{\binom{\Lambda}{t}} \right), \quad \forall t \in [0 : \Lambda]
\end{equation}
and where memory sharing is used between any two consecutive integer values of $t$. In~\Cref{sec: Converse Proof of MACC with Combinatorial Topology} this performance will be shown to be optimal under the assumption of uncoded placement.

\begin{remark}
    The coding scheme and its performance fully incorporate the following known scenarios.
    \begin{enumerate}
        \item Case $\bm{K}_{\text{comb}} = (K_0, 0, \dots, 0)$. This scenario corresponds to the case where there are $K_0$ users and none of them is connected to any of the $\Lambda$ caches. In this case, uncoded delivery is optimal and the load performance is equal to $K_0$, independently of the memory value, as expected.
        \item Case $\bm{K}_{\text{comb}} = (0, 1, 0, \dots, 0)$. This scenario corresponds to the well-known dedicated-caches case where there is only one user connected to any one of the $\Lambda$ caches. This implies that there are $K = \Lambda$ users in total and this case corresponds to the standard MAN setting, for which the MAN scheme was shown in~\cite{Wan2020AnPlacement} (see also~\cite{Yu2018ThePrefetching}) to be optimal under uncoded cache placement.
        \item Case $\bm{K}_{\text{comb}} = (0, K_1, 0, \dots, 0)$ with $K_1 > 1$. This scenario corresponds to the shared-caches setting with uniform user-to-cache association profile, where the corresponding achievable load was also shown to be optimal (cf.~\cite{Parrinello2020FundamentalPrefetching}).
        \item Case $\bm{K}_{\text{comb}} = (0, \dots, 0, K_\lambda, 0, \dots, 0)$ with $K_\lambda = 1$ for some $\lambda \in [2 : \Lambda]$. As already mentioned, this scenario corresponds to the setting considered in~\cite{Muralidhar2021Maddah-Ali-NiesenCaching}. No optimality result was stated.\qed
    \end{enumerate}
\end{remark}

\section{Converse Proof of \texorpdfstring{\Cref{thm: MACC with Generalized Combinatorial Topology}}{Theorem~\ref{thm: MACC with Generalized Combinatorial Topology}}}\label{sec: Converse Proof of MACC with Combinatorial Topology}

The converse relies on the well-known acyclic subgraph index coding bound, which has been extensively used in various other settings (see for example~\cite{Wan2020AnPlacement, Parrinello2020FundamentalPrefetching, Brunero2021UnselfishCaching} to name a few) in order to derive lower bounds on the optimal worst-case load in caching under the assumption of uncoded prefetching. To clarify the connection between this bound and our setting, we first provide a brief presentation of the index coding problem and its connection to coded caching. 

\subsection{Definition of the Index Coding Problem}

The index coding problem \cite{Bar-Yossef2011IndexInformation, Arbabjolfaei2018FundamentalsCoding, Thapa2017InterlinkedCliques, Vaddi2018OptimalStructure, Arbabjolfaei2013OnCoding} consists of a central server having access to $N'$ independent messages as well as consists of $K'$ users that are connected to the server via a shared error-free broadcast channel. Each user $k \in [K']$ has a set of desired messages $\mathcal{M}_k \subseteq [N']$, which is called the \emph{desired message set}, while also having access to another subset of messages $\mathcal{A}_k \subseteq [N']$, which is called the \emph{side information set}. To avoid trivial scenarios, it is commonly assumed that $\mathcal{M}_k \neq \emptyset$, $\mathcal{A}_k \neq [N']$ and $\mathcal{M}_k \cap \mathcal{A}_k = \emptyset$ for each $k \in [K']$.

The index coding problem is usually described in terms of its \emph{side information graph}. Let $M_i$ be the $i$-th message in the set $[N']$. Then, such graph is a directed graph where each vertex is a desired message and where there exists an edge from a desired message $M_i$ to a desired message $M_j$ if and only if message $M_i$ is in the side information set of the user requesting message $M_j$. The derivation of our converse is based on the following bound from \cite[Corollary 1]{Arbabjolfaei2013OnCoding}.

\begin{lemma}[{\cite[Corollary 1]{Arbabjolfaei2013OnCoding}}]\label{lem: Acyclic Subgraph Converse Bound}
    Consider an index coding problem with $N'$ messages $M_i$ for $i \in [N']$. The minimum number of transmitted bits $\rho$ is lower bounded as
    \begin{equation}
        \rho \geq \sum_{i \in \mathcal{J}}|M_i|
    \end{equation}
    for any acyclic subgraph $\mathcal{J}$ of the problem's side information graph.
\end{lemma}

\subsection{Main Proof of the Converse in~\texorpdfstring{\Cref{thm: MACC with Generalized Combinatorial Topology}}{Theorem~\ref{thm: MACC with Generalized Combinatorial Topology}}}

The first step in our converse proof consists of dividing, in the most generic manner, each file into a maximum of $2^\Lambda$ disjoint subfiles as
\begin{equation}\label{eqn: File Splitting}
    W_{n} = \left\{W_{n, \mathcal{T}} : \mathcal{T} \subseteq [\Lambda] \right\}, \quad \forall n \in [N]
\end{equation}
where we identify with $W_{n, \mathcal{T}}$ the subfile which is exclusively stored by the caches in $\mathcal{T}$. Such placement is designated as \emph{uncoded} because the bits of the library files are simply copied within the caches.

\subsubsection{Constructing the Index Coding Bound}

Assuming that each user requests a distinct\footnote{Notice that the set of worst-case demands may not include the set of demand vectors $\bm{d}$ with all distinct entries. However, this is not a problem, since our goal is to derive a converse bound on the worst-case load. Indeed, our choice of treating distinct demands yields a converse bound, which --- while it does not need to be, a priori, the tightest bound --- is a valid bound. In our case, the bound proves to be tight.} file, we consider the index coding problem with $K' = K = \sum_{\lambda = 0}^{\Lambda} K_\lambda \binom{\Lambda}{\lambda}$ users and $N' = \sum_{\lambda = 0}^{\Lambda} K_\lambda \binom{\Lambda}{\lambda} 2^{\Lambda - \lambda}$ independent messages, where each such message represents a subfile requested by some user (who naturally does not have access to it via a cache). Recalling that $W_{d_{\mathcal{U}_{k}}}$ denotes the file requested by the user identified by $\mathcal{U}_{k}$, the desired message set and the side information set are respectively given, in their most generic form, by
\begin{align}
    \mathcal{M}_{\mathcal{U}_k} & = \{W_{d_{\mathcal{U}_{k}}, \mathcal{T}} : \mathcal{T} \subseteq [\Lambda] \setminus \mathcal{U}\} \\
    \mathcal{A}_{\mathcal{U}_k} & = \{W_{n, \mathcal{T}} : n \in [N], \mathcal{T} \subseteq [\Lambda], \mathcal{T} \cap \mathcal{U} \neq \emptyset\}
\end{align}
for each user $\mathcal{U}_k$ with $k \in [K_{|\mathcal{U}|}]$ and $\mathcal{U} \subseteq [\Lambda]$. Here, the side information graph consists of a directed graph where each vertex is a subfile, and where there is an edge from the subfile $W_{d_{\mathcal{P}_{k_1}}}$ to the subfile $W_{d_{\mathcal{Q}_{k_2}}}$ if and only if $W_{d_{\mathcal{P}_{k_1}}} \in \mathcal{A}_{\mathcal{Q}_{k_2}}$ with $\mathcal{P} \subseteq [\Lambda]$, $\mathcal{Q} \subseteq [\Lambda]$, $k_1 \in [K_{|\mathcal{P}|}]$, $k_2 \in [K_{|\mathcal{Q}|}]$ and $\mathcal{P}_{k_1} \neq \mathcal{Q}_{k_2}$. Since our aim is to apply \Cref{lem: Acyclic Subgraph Converse Bound}, we need to consider acyclic sets of vertices $\mathcal{J}$ in the side information graph. Toward this, we take advantage of the following lemma, which holds for any connectivity $b \in \mathcal{B}$.

\begin{lemma}\label{lem: Acyclic Lemma}
    Let $\bm{d} = (\bm{d}_0, \dots, \bm{d}_\Lambda)$ be a demand vector and let $\bm{c} = (c_1, \dots, c_\Lambda)$ be a permutation of the $\Lambda$ caches. The following set of vertices
    \begin{equation}
        \bigcup_{k \in [K_{\emptyset, b}]} \bigcup_{\mathcal{T} \subseteq [\Lambda]} \left\{ W_{d_{\emptyset_k, \mathcal{T}}} \right\} \cup \bigcup_{\lambda \in [\Lambda]}  \bigcup_{i \in [\lambda : \Lambda]}  \bigcup_{\substack{\mathcal{U}^i \subseteq \{c_1, \dots, c_i\} : |\mathcal{U}^i| = \lambda,\\ c_i \in \mathcal{U}^i}}  \bigcup_{k \in [K_{\mathcal{U}^i, b}]}  \bigcup_{\mathcal{T}_i \subseteq [\Lambda] \setminus \{c_1, \dots, c_i\}} \left\{ W_{d_{\mathcal{U}^i_k, \mathcal{T}_i}} \right\}
    \end{equation}
    is acyclic for any connectivity $b \in \mathcal{B}$.
\end{lemma}

\begin{proof}
    The proof is reported in~\refappendix{app: Proof of Acyclic Lemma}.
\end{proof}

Consider a demand vector $\bm{d} = (\bm{d}_0, \dots, \bm{d}_\Lambda)$ and a permutation $\bm{c} = (c_1, \dots, c_\Lambda)$ of the set $[\Lambda]$. If we specialize the acyclic set presented in~\Cref{lem: Acyclic Lemma} to the generalized combinatorial topology, then applying \Cref{lem: Acyclic Subgraph Converse Bound} yields the following lower bound
\begin{equation}\label{eqn: Index Coding Lower Bound 1}
    BR^\star_{\text{comb}} \geq \sum_{k \in [K_0]} \sum_{\mathcal{T} \subseteq [\Lambda]} \left| W_{d_{\emptyset_k, \mathcal{T}}} \right| + \sum_{\lambda \in [\Lambda]} \sum_{i \in [\lambda : \Lambda]} \sum_{\substack{\mathcal{U}^i \subseteq \{c_1, \dots, c_i\} : |\mathcal{U}^i| = \lambda,\\ c_i \in \mathcal{U}^i}} \sum_{k \in [K_\lambda]} \sum_{\mathcal{T}_i \subseteq [\Lambda] \setminus \{c_1, \dots, c_i\}} \left| W_{d_{\mathcal{U}^i_k, \mathcal{T}_i}} \right|.
\end{equation}

\subsubsection{Constructing the Optimization Problem}

Now our goal is to create several bounds as the one in~\eqref{eqn: Index Coding Lower Bound 1} considering any vector $\bm{d} \in \mathcal{D}$ and any vector $\bm{c} \in \mathcal{C}$, where we denote by $\mathcal{D}$ and $\mathcal{C}$ the set of possible demand vectors with distinct entries and the set of possible permutation vectors of the set $[\Lambda]$, respectively. Our aim is then to average all these bounds to obtain in the end a useful lower bound on the optimal worst-case load. Considering that $|\mathcal{D}| = \binom{N}{K}K!$ and $|\mathcal{C}| = \Lambda!$, we aim to simplify the expression given by
\begin{equation}\label{eqn: Complete Lower Bound 1}
    \begin{split}
        \binom{N}{K}K!\Lambda! BR^\star_{\text{comb}} & \geq \sum_{\bm{d} \in \mathcal{D}} \sum_{\bm{c} \in \mathcal{C}} \Bigg(\sum_{k \in [K_0]} \sum_{\mathcal{T} \subseteq [\Lambda]} \left| W_{d_{\emptyset_k, \mathcal{T}}} \right| + \\ 
                                                      & + \sum_{\lambda \in [\Lambda]} \sum_{i \in [\lambda : \Lambda]} \sum_{\substack{\mathcal{U}^i \subseteq \{c_1, \dots, c_i\} : |\mathcal{U}^i| = \lambda, \\ c_i \in \mathcal{U}^i}} \sum_{k \in [K_\lambda]} \sum_{\mathcal{T}_i \subseteq [\Lambda] \setminus \{c_1, \dots, c_i\}} \left| W_{d_{\mathcal{U}^i_k, \mathcal{T}_i}} \right|\Bigg).
    \end{split}
\end{equation}
The next step consists of simplifying the expression in~\eqref{eqn: Complete Lower Bound 1}. Toward simplifying, we count how many times each subfile $W_{n, \mathcal{T}}$ --- for any given $n \in [N]$, $\mathcal{T} \subseteq [\Lambda]$ and $|\mathcal{T}| = t$ for some $t \in [0 : \Lambda]$ --- appears in~\eqref{eqn: Complete Lower Bound 1}.

Assume that the file $W_n$ is demanded by user $\emptyset_{k}$ for some $k \in [K_0]$. Out of the entire set $\mathcal{D}$ of all possible distinct demands, we find a total of $\binom{N}{K}K!/N$ distinct demands for which a file is requested by the same user. Hence, since there are $\binom{N}{K}K!/N$ distinct demands for which the file $W_{n}$ is requested by user $\emptyset_{k}$, the subfile $W_{n, \mathcal{T}}$ is counted a total of $\Lambda!\binom{N}{K}K!/N$ times within the set of demands for which such file is requested by user $\emptyset_{k}$. The $\Lambda!$ term comes from the fact that the set
\begin{equation}
    \bigcup_{k \in [K_0]} \bigcup_{\mathcal{T} \subseteq [\Lambda]} \{W_{d_{\emptyset_k, \mathcal{T}}}\}
\end{equation}
does not depend on the permutation vector $\bm{c}$, thus the subfile $W_{n, \mathcal{T}}$ appears in such set independently of the vector $\bm{c}$. Since there are $\Lambda!$ such vectors, the subfile $W_{n, \mathcal{T}}$ is counted $\Lambda!$ times any time it is requested by the user $\emptyset_{k}$. The same reasoning follows for each $k \in [K_0]$, hence we can conclude that the subfile $W_{n, \mathcal{T}}$ is counted
\begin{equation}
    K_0\Lambda!\frac{\binom{N}{K}K!}{N}
\end{equation}
times when we span all cases of distinct demand vectors for which the file $W_n$ is requested by the set of users connected to $0$ caches.

Assume now that the file $W_n$ is demanded by user $\mathcal{U}_{k}$ for some $k \in [K_1]$ and for some $\mathcal{U} \subseteq [\Lambda] \setminus \mathcal{T}$ with $|\mathcal{U}| = 1$. Recall also that such file is requested by user $\mathcal{U}_{k}$ a total of $\binom{N}{K}K!/N$ times. Then, within the set of demands for which such user requests the file $W_n$, the subfile $W_{n, \mathcal{T}}$ is counted only when the elements in the set $\mathcal{U}$ appear in the vector $\bm{c}$ before\footnote{Indeed, the subfile $W_{n, \mathcal{T}}$ appears in the acyclic graph chosen as in~\Cref{lem: Acyclic Lemma} for all those permutations $\bm{c} = (c_1, \dots, c_\Lambda)$ for which $\mathcal{U} \subseteq \{c_1, \dots, c_i\}$ and $\mathcal{T} \subseteq \{c_{i + 1}, \dots, c_{\Lambda}\}$, i.e., this happens whenever the elements in $\mathcal{T}$ are after the elements in $\mathcal{U}$ in the permutation vector $\bm{c}$.} the elements in the set $\mathcal{T}$. Since there is a total of $t!(\Lambda - 1 - t)!\binom{\Lambda}{t + 1}$ such vectors $\bm{c}$ in the set $\mathcal{C}$, the subfile $W_{n, \mathcal{T}}$ is counted a total of $t!(\Lambda - 1 - t)!\binom{\Lambda}{t + 1}\binom{N}{K}K!/N$ times within the set of demands for which user $\mathcal{U}_{k}$ requests the file $W_n$. The same reasoning follows for each $k \in [K_1]$ and for each $\mathcal{U} \subseteq [\Lambda] \setminus \mathcal{T}$ with $|\mathcal{U}| = 1$, hence the subfile $W_{n, \mathcal{T}}$ is counted a total of 
\begin{equation}
    K_1 (\Lambda - t) t!(\Lambda - 1 - t)!\binom{\Lambda}{t + 1}\frac{\binom{N}{K}K!}{N}
\end{equation}
times across all the demands for which the file $W_n$ is requested by the set of users connected to $1$ cache.

Let us consider now that the file $W_n$ is demanded by user $\mathcal{U}_{k}$ for some $k \in [K_\lambda]$ and for some $\mathcal{U} \subseteq [\Lambda] \setminus \mathcal{T}$ with $|\mathcal{U}| = \lambda$. Recall also that such file is requested by user $\mathcal{U}_{k}$ a total of $\binom{N}{K}K!/N$ times. Then, within the set of demands for which such user requests the file $W_n$, the subfile $W_{n, \mathcal{T}}$ is counted only when the elements in the set $\mathcal{U}$ appear in the vector $\bm{c}$ before the elements in the set $\mathcal{T}$. Since there is a total of $\lambda!t!(\Lambda - \lambda - t)!\binom{\Lambda}{t + \lambda}$ such vectors $\bm{c}$, the subfile $W_{n, \mathcal{T}}$ is counted a total of $\lambda!t!(\Lambda - \lambda - t)!\binom{\Lambda}{t + \lambda}\binom{N}{K}K!/N$ times within the set of demands for which user $\mathcal{U}_{k}$ requests the file $W_n$. The same reasoning follows for each $k \in [K_\lambda]$ and for each $\mathcal{U} \subseteq [\Lambda] \setminus \mathcal{T}$ with $|\mathcal{U}| = \lambda$, hence the subfile $W_{n, \mathcal{T}}$ is counted a total of
\begin{equation}
     K_\lambda \binom{\Lambda - t}{\lambda} \lambda!t!(\Lambda - \lambda - t)!\binom{\Lambda}{t + \lambda} \frac{\binom{N}{K}K!}{N}
\end{equation}
times within the set of demands for which the file $W_n$ is requested by the set of users connected to $\lambda$ caches.

Consequently, if we consider all distinct demands in the set $\mathcal{D}$, the subfile $W_{n, \mathcal{T}}$ is counted a total of
\begin{equation}
    \sum_{\lambda = 0}^{\Lambda - t} K_\lambda \binom{\Lambda - t}{\lambda} \lambda!t!(\Lambda - \lambda - t)!\binom{\Lambda}{t + \lambda} \frac{\binom{N}{K}K!}{N}
\end{equation}
times, which gives us the number of times this same subfile appears in~\eqref{eqn: Complete Lower Bound 1}. The same reasoning follows for any $n \in [N]$ and for any $\mathcal{T} \subseteq [\Lambda]$ with $|\mathcal{T}| = t$. Thus, the expression in~\eqref{eqn: Complete Lower Bound 1} can be rewritten as
\begin{align}
    R^\star_{\text{comb}} & \geq \frac{1}{\binom{N}{K}K!\Lambda!} \sum_{t = 0}^{\Lambda} \sum_{\lambda = 0}^{\Lambda - t} K_\lambda \binom{\Lambda - t}{\lambda} \lambda!t!(\Lambda - \lambda - t)!\binom{\Lambda}{t + \lambda} \binom{N}{K}K! x_t \\
                          & = \sum_{t = 0}^{\Lambda} \sum_{\lambda = 0}^{\Lambda - t} K_\lambda \frac{\binom{\Lambda}{t + \lambda}}{\binom{\Lambda}{t}} x_t \\
                          & = \sum_{t = 0}^{\Lambda} f(t) x_t
\end{align}
where we define
\begin{align}
    f(t) & \coloneqq \sum_{\lambda = 0}^{\Lambda - t} K_\lambda \frac{\binom{\Lambda}{t + \lambda}}{\binom{\Lambda}{t}} \\
    0 \leq x_t & \coloneqq \sum_{n \in [N]} \sum_{\mathcal{T} \subseteq [\Lambda] : |\mathcal{T}| = t} \frac{\left| W_{n, \mathcal{T}} \right|}{NB}.
\end{align}

At this point, we seek to lower bound the minimum worst-case load $R^\star_{\text{comb}}$ by solving the following optimization problem
\begin{subequations}\label{eqn: Optimization Problem 1}
  \begin{alignat}{2}
    & \min_{\bm{x}}  & \quad & \sum_{t = 0}^{\Lambda} f(t) x_{t} \\
    & \text{subject to}  &  & \sum_{t = 0}^{\Lambda} x_{t} = 1 \label{eqn: File-Size Constraint 1} \\
    & & & \sum_{t = 0}^{\Lambda} t x_{t} \leq \frac{\Lambda M}{N} \label{eqn: Memory-Size Constraint 1}
  \end{alignat}
\end{subequations}
where \eqref{eqn: File-Size Constraint 1} and \eqref{eqn: Memory-Size Constraint 1} correspond to the file-size constraint and the cumulative cache-size constraint, respectively.

\subsubsection{Solving the Optimization Problem}

Since the auxiliary variable $x_{t}$ can be considered as a probability mass function, the optimization problem in \eqref{eqn: Optimization Problem 1} can be seen as the minimization of $\mathbb{E}[f(t)]$. Moreover, the following holds.
\begin{lemma}\label{lem: Strictly Decreasing Sequence}
  The function $f(t)$ is convex and decreasing in $t$.
\end{lemma}
\begin{proof}
  The proof is reported in \refappendix{app: Proof of the Strictly Decreasing Sequence Lemma}.
\end{proof}
Taking advantage of \Cref{lem: Strictly Decreasing Sequence}, we can write $\mathbb{E}[f(t)] \geq f(\mathbb{E}[t])$ using Jensen's inequality. Then, since $f(t)$ is decreasing with increasing $t \in [0 : \Lambda]$, we can further write $f(\mathbb{E}[t]) \geq f(\Lambda M/N)$ taking advantage of the fact that $\mathbb{E}[t]$ is upper bounded as in \eqref{eqn: Memory-Size Constraint 1}. Consequently, $\mathbb{E}[f(t)] \geq f(\Lambda M/N)$, and thus for $t = \Lambda M/N$ the optimal worst-case load $R^\star_{\text{comb}}$ is lower bounded by $R_{\text{comb}, \text{LB}}$ which is a piece-wise linear curve with corner points
\begin{equation}\label{eqn: Lower Bound}
  (M, R_{\text{comb}, \text{LB}}) = \left(t\frac{N}{\Lambda}, \sum_{\lambda = 0}^{\Lambda - t}K_\lambda \frac{\binom{\Lambda}{t + \lambda}}{\binom{\Lambda}{t}} \right), \quad \forall t \in [0 : \Lambda].
\end{equation}
Since $R_{\text{comb}, \text{LB}} = R_{\text{comb}, \text{UB}}$, we can state that the optimal worst-case load $R^\star_{\text{comb}}$ under uncoded placement is a piece-wise linear curve with the same corner points as in~\eqref{eqn: Lower Bound}. This concludes the proof.\qed

\section{Proof of \texorpdfstring{\Cref{thm: Converse on Optimal Average Worst-Case Load}}{Theorem~\ref{thm: Converse on Optimal Average Worst-Case Load}}}\label{sec: Proof of Converse on Optimal Average Worst-Case Load}

The proof follows along the lines of the index coding approach presented in~\Cref{sec: Converse Proof of MACC with Combinatorial Topology}. We remind that we are considering now the scenario where there are $K$ users and each of them is connected to exactly $\lambda$ caches for a fixed value of $\lambda \in [\Lambda]$, where the connectivity is not constrained by further restrictions and where the connectivity is not known in advance (i.e., it is not known during placement). We remind the reader that we are interested in the connectivity ensemble $\mathcal{B}_{\lambda}$ as a whole, and that we want to develop a converse on the optimal average worst-case load $R^{\star}_{\text{avg}, \mathcal{B}_\lambda}$ under uncoded and fixed cache placement. Furthermore, we recall that we assume each of these connectivities to be equiprobable and that for any $b \in \mathcal{B}_\lambda$ the number of users takes the form
\begin{equation}
    K = K'_\lambda = \sum_{\mathcal{U} \subseteq [\Lambda] : |\mathcal{U}| = \lambda} K_{\mathcal{U}, b}.
\end{equation}

\subsection{Constructing the Index Coding Bound}\label{sec: Constructing the Index Coding Bound}

The first step of the proof consists of dividing generically --- and independently of the connectivity $b \in \mathcal{B}_\lambda$ --- each file into a maximum of $2^\Lambda$ non-overlapping subfiles as in~\eqref{eqn: File Splitting}, recalling that $W_{n, \mathcal{T}}$ represents the subfile which is cached uniquely and exactly by the caches in $\mathcal{T}$. Then, similarly to how we proceeded in~\Cref{sec: Converse Proof of MACC with Combinatorial Topology}, we always assume the demand vector $\bm{d} = (0, \dots, 0, \bm{d}_\lambda, 0, \dots, 0)$ to have distinct entries. Given a connectivity $b \in \mathcal{B}_\lambda$, we consider the index coding problem with $K' = K$ users and $N' = K 2^{\Lambda - \lambda}$ independent messages. Notice that, since each user is connected to exactly $\lambda$ caches, the number of desired subfiles is exactly equal to $2^{\Lambda - \lambda}$, hence $N'$ is the total number of subfiles requested by the $K$ users. Recalling that $W_{d_{\mathcal{U}_{k}}}$ denotes the file requested by the user identified by $\mathcal{U}_{k}$, the desired message set and the side information set are respectively given by
\begin{align}
    \mathcal{M}_{\mathcal{U}_k} & = \{W_{d_{\mathcal{U}_{k}}, \mathcal{T}} : \mathcal{T} \subseteq [\Lambda] \setminus \mathcal{U}\} \\
    \mathcal{A}_{\mathcal{U}_k} & = \{W_{n, \mathcal{T}} : n \in [N], \mathcal{T} \subseteq [\Lambda], \mathcal{T} \cap \mathcal{U} \neq \emptyset\}
\end{align}
for each user $\mathcal{U}_k$ with $\mathcal{U} \subseteq [\Lambda]$, $|\mathcal{U}| = \lambda$ and $k \in [K_{\mathcal{U}, b}]$. The side information graph consists again of a directed graph where each vertex is a subfile, and where there is a connection from the subfile $W_{d_{\mathcal{P}_{k_1}}}$ to the subfile $W_{d_{\mathcal{Q}_{k_2}}}$ if and only if $W_{d_{\mathcal{P}_{k_1}}} \in \mathcal{A}_{\mathcal{Q}_{k_2}}$ with $\mathcal{P} \subseteq [\Lambda]$, $\mathcal{Q} \subseteq [\Lambda]$, $|\mathcal{P}| = |\mathcal{Q}| = \lambda$, $k_1 \in [K_{\mathcal{P}, b}]$, $k_2 \in [K_{\mathcal{Q}, b}]$ and $\mathcal{P}_{k_1} \neq \mathcal{Q}_{k_2}$. Since our aim is to apply \Cref{lem: Acyclic Subgraph Converse Bound}, we need to consider acyclic sets of vertices $\mathcal{J}$ in the side information graph. In the spirit of~\Cref{lem: Acyclic Lemma}, we know that now the set
\begin{equation}
    \bigcup_{i \in [\lambda : \Lambda]} \bigcup_{\substack{\mathcal{U}^i \subseteq \{c_1, \dots, c_i\} : |\mathcal{U}^i| = \lambda,\\ c_i \in \mathcal{U}^i}} \bigcup_{k \in [K_{\mathcal{U}^i, b}]} \bigcup_{\mathcal{T}_i \subseteq [\Lambda] \setminus \{c_1, \dots, c_i\}} \left\{ W_{d_{\mathcal{U}^i_k, \mathcal{T}_i}} \right\}
\end{equation}
is acyclic for any demand vector $\bm{d}$ and for any permutation of the $\Lambda$ caches represented by the vector $\bm{c} = (c_1, \dots, c_\Lambda)$. Applying \Cref{lem: Acyclic Subgraph Converse Bound} yields the following lower bound
\begin{equation}\label{eqn: Index Coding Lower Bound 2}
    BR^\star_{b} \geq \sum_{i \in [\lambda : \Lambda]} \sum_{\substack{\mathcal{U}^i \subseteq \{c_1, \dots, c_i\} : |\mathcal{U}^i| = \lambda,\\ c_i \in \mathcal{U}^i}} \sum_{k \in [K_{\mathcal{U}^i, b}]} \sum_{\mathcal{T}_i \subseteq [\Lambda] \setminus \{c_1, \dots, c_i\}} \left| W_{d_{\mathcal{U}^i_k, \mathcal{T}_i}}\right|
\end{equation}
where the term $R^\star_{b}$ represents the optimal worst-case load given the connectivity $b \in \mathcal{B}_\lambda$.

\subsection{Counting the Connectivities}

As it will be of use later, we proceed to count how many possible connectivities exist in $\mathcal{B}_\lambda$, recalling that such ensemble includes connectivities for which each user connects to exactly $\lambda$ caches. Toward this, if we let $x_\mathcal{U}$ be a non-negative integer value that represents the total number of users connected to the caches in $\mathcal{U}$, then the number of connectivities in $\mathcal{B}_\lambda$ is equal to the number of non-negative integer solutions of the equation
\begin{equation}
    \sum_{\mathcal{U} \subseteq [\Lambda] : |\mathcal{U}| = \lambda}x_\mathcal{U} = K.
\end{equation}
This is simply equal to the number of ways we can express the integer $K$ as the sum of a sequence of $\binom{\Lambda}{\lambda}$ non-negative integers. Hence, such number corresponds to the number of $\binom{\Lambda}{\lambda}$-weak compositions~\cite{Stanley2011EnumerativeCombinatorics} of the integer $K$ and is given by
\begin{equation}
    |\mathcal{B}_\lambda| = \binom{K + \binom{\Lambda}{\lambda} - 1}{K}.
\end{equation}
This highlights the fact that a connectivity $b \in \mathcal{B}_\lambda$ is nothing but a specific way to distribute $K$ users among $\binom{\Lambda}{\lambda}$ possible \emph{states}, where a \emph{state} represents a set of $\lambda$ caches which a user can be connected to.

\subsection{Constructing the Optimization Problem}

Our goal is to develop an information-theoretic converse on the optimal average worst-case load $R^\star_{\text{avg}, \mathcal{B}_\lambda}$ over the ensemble of connectivities $\mathcal{B}_\lambda$. As it will be of use later, we point out that the cache placement minimizing the average worst-case load does not necessarily minimize the worst-case load of each connectivity $b \in \mathcal{B}_\lambda$, which in turn implies that
\begin{equation}\label{eqn: Average Worst-Case Inequality 1}
    R^\star_{\text{avg}, \mathcal{B}_\lambda} \geq \frac{1}{\left| \mathcal{B}_\lambda \right|} \sum_{b \in \mathcal{B}_\lambda} R^\star_{b}.
\end{equation}

To proceed, we split the ensemble of interest as $\mathcal{B}_\lambda = \mathcal{B}_{\lambda, \text{a}} \cup \mathcal{B}_{\lambda, \text{b}}$, where $\mathcal{B}_{\lambda, \text{a}}$ is the set of connectivities for which all the $K$ users are connected to the same $\lambda$ caches and where $\mathcal{B}_{\lambda, \text{b}} = \mathcal{B}_{\lambda} \setminus \mathcal{B}_{\lambda, \text{a}}$. For each connectivity $b \in \mathcal{B}_{\lambda, \text{a}}$, we will create several bounds as the one in~\eqref{eqn: Index Coding Lower Bound 2}. To do so, we consider any vector $\bm{d} \in \mathcal{D}$ and for each such vector we employ $\Lambda!$ times the permutation vector $\bm{c}_{b}$ whose first $\lambda$ positions describe the indices of the caches to which the $K$ users are connected\footnote{For example, let $\Lambda = 7$ and $\lambda = 2$. If we consider the connectivity $b_1 \in \mathcal{B}_{\lambda, \text{a}}$ for which all $K$ users are connected to the caches $\{3, 5\}$, then we consider the permutation vector $\bm{c}_{b_1} = (3, 5, 1, 2, 4, 6, 7)$. Similarly, if we consider the connectivity $b_2 \in \mathcal{B}_{\lambda, \text{a}}$ for which all $K$ users are connected to the caches $\{4, 7\}$, then we consider the permutation vector $\bm{c}_{b_2} = (4, 7, 1, 2, 3, 5, 6)$. For each connectivity $b \in \mathcal{B}_{\lambda, \text{a}}$, the vector $\bm{c}_{b}$ is employed $\Lambda!$ times for each demand vector $\bm{d} \in \mathcal{D}$.} according to the connectivity $b \in \mathcal{B}_{\lambda, \text{a}}$. We assume that the first $\lambda$ elements in $\bm{c}_b$ are put in ascending order and then, similarly, the remaining $\Lambda - \lambda$ elements of the vector $\bm{c}_b$ are also placed in ascending order. Instead, for each connectivity $b \in \mathcal{B}_{\lambda, \text{b}}$, we will create several bounds as the one in~\eqref{eqn: Index Coding Lower Bound 2} by considering any demand vector $\bm{d} \in \mathcal{D}$ and any permutation vector $\bm{c} \in \mathcal{C}$. Our aim is then to average all these bounds to eventually obtain a useful lower bound on the optimal average worst-case load. Considering that we have $|\mathcal{D}| = \binom{N}{K}K!$ and $|\mathcal{C}| = \Lambda !$, we aim to simplify the expression given by
\begin{equation}
    \begin{split}
        \binom{N}{K}K!\Lambda! B \sum_{b \in \mathcal{B}_\lambda} R^\star_b & \geq \sum_{b \in \mathcal{B}_{\lambda, \text{a}}} \sum_{\bm{d} \in \mathcal{D}} \Lambda! \sum_{k \in [K]} \sum_{\mathcal{T}_i \subseteq [\Lambda] \setminus \{c_{b, 1}, \dots, c_{b, \lambda}\}} \left| W_{d_{\{c_{b, 1}, \dots, c_{b, \lambda}\}_k, \mathcal{T}_i}}\right| + \\
                                                                    & + \sum_{b \in \mathcal{B}_{\lambda, \text{b}}} \sum_{\bm{d} \in \mathcal{D}} \sum_{\bm{c} \in \mathcal{C}} \sum_{i \in [\lambda : \Lambda]} \sum_{\substack{\mathcal{U}^i \subseteq \{c_1, \dots, c_i\} : |\mathcal{U}^i| = \lambda,\\ c_i \in \mathcal{U}^i}} \sum_{k \in [K_{\mathcal{U}^i, b}]} \sum_{\mathcal{T}_i \subseteq [\Lambda] \setminus \{c_1, \dots, c_i\}} \left| W_{d_{\mathcal{U}^i_k, \mathcal{T}_i}}\right|    
    \end{split}
\end{equation}
which can be rewritten by means of~\eqref{eqn: Average Worst-Case Inequality 1} as
\begin{equation}\label{eqn: Complete Lower Bound 2}
    \begin{split}
        R^\star_{\text{avg}, \mathcal{B}_\lambda} & \geq \frac{1}{\binom{N}{K}K!\Lambda! \left| \mathcal{B}_\lambda \right| B } \Bigg( \sum_{b \in \mathcal{B}_{\lambda, \text{a}}} \sum_{\bm{d} \in \mathcal{D}} \Lambda! \sum_{k \in [K]} \sum_{\mathcal{T}_i \subseteq [\Lambda] \setminus \{c_{b, 1}, \dots, c_{b, \lambda}\}} \left| W_{d_{\{c_{b, 1}, \dots, c_{b, \lambda}\}_k, \mathcal{T}_i}}\right| + \\
                                                                & + \sum_{b \in \mathcal{B}_{\lambda, \text{b}}} \sum_{\bm{d} \in \mathcal{D}} \sum_{\bm{c} \in \mathcal{C}} \sum_{i \in [\lambda : \Lambda]} \sum_{\substack{\mathcal{U}^i \subseteq \{c_1, \dots, c_i\} : |\mathcal{U}^i| = \lambda,\\ c_i \in \mathcal{U}^i}} \sum_{k \in [K_{\mathcal{U}^i, b}]} \sum_{\mathcal{T}_i \subseteq [\Lambda] \setminus \{c_1, \dots, c_i\}} \left| W_{d_{\mathcal{U}^i_k, \mathcal{T}_i}}\right| \Bigg).
    \end{split}
\end{equation}
Our next step is to simplify the expression in~\eqref{eqn: Complete Lower Bound 2} and part of doing so involves counting how many times each subfile $W_{n, \mathcal{T}}$ --- for any $n \in [N]$, $\mathcal{T} \subseteq [\Lambda]$ and $|\mathcal{T}| = t$ for some $t \in [0 : \Lambda]$ --- appears in~\eqref{eqn: Complete Lower Bound 2}.

Consider the connectivities in $\mathcal{B}_{\lambda, \text{a}}$. Assume that the file $W_n$ is demanded by the user $\mathcal{U}_k$ for some $k \in [K]$ and for some $\mathcal{U} \subseteq [\Lambda] \setminus \mathcal{T}$ with $|\mathcal{U}| = \lambda$. Since this corresponds to the connectivity $b \in \mathcal{B}_{\lambda, \text{a}}$ for which all the users are associated to the caches in $\mathcal{U}$, we employ $\Lambda!$ consecutive times the permutation vector $\bm{c}_{b}$ having in the first $\lambda$ positions the elements in $\mathcal{U}$. Consequently, since here the elements in $\mathcal{T}$ are by construction after the elements in $\mathcal{U}$ in the vector $\bm{c}_{b}$, we can deduce that the subfile $W_{n, \mathcal{T}}$ is counted $\Lambda!$ times --- namely, once for each of the $\Lambda!$ times that the vector $\bm{c}_{b}$ is employed --- whenever the file $W_n$ is requested by user $\mathcal{U}_k$. Recalling that any file is requested by user $\mathcal{U}_{k}$ a total of $\binom{N}{K}K!/N$ times, the subfile $W_{n, \mathcal{T}}$ is counted $\Lambda!\binom{N}{K}K!/N$ within the set of demands for which user $\mathcal{U}_k$ requests the file $W_n$. The same reasoning follows for each $k \in [K]$ and for each $\mathcal{U} \subseteq [\Lambda] \setminus \mathcal{T}$ with $|\mathcal{U}| = \lambda$, hence the subfile $W_{n, \mathcal{T}}$ is counted
\begin{equation}
    K\binom{\Lambda - t}{\lambda}\Lambda!\frac{\binom{N}{K}K!}{N}
\end{equation}
times when we focus on connectivities in $\mathcal{B}_{\lambda, \text{a}}$.

Consider now a specific connectivity $b \in \mathcal{B}_{\lambda, \text{b}}$. Assume that the file $W_n$ is demanded by the user $\mathcal{U}_{k}$ for some $k \in [K_{\mathcal{U}, b}]$ and for some $\mathcal{U} \subseteq [\Lambda] \setminus \mathcal{T}$ with $|\mathcal{U}| = \lambda$, and recall once more that such file is requested by user $\mathcal{U}_{k}$ a total of $\binom{N}{K}K!/N$ times. Within the set of demand vectors for which user $\mathcal{U}_{k}$ requests the file $W_n$, the subfile $W_{n, \mathcal{T}}$ is counted only when the elements in the set $\mathcal{U}$ appear in the vector $\bm{c}$ before the elements in the set $\mathcal{T}$. Since there is a total of $\lambda!t!(\Lambda - \lambda - t)!\binom{\Lambda}{t + \lambda}$ such vectors $\bm{c}$, the subfile $W_{n, \mathcal{T}}$ is counted a total of $\lambda!t!(\Lambda - \lambda - t)!\binom{\Lambda}{t + \lambda}\binom{N}{K}K!/N$ times within the set of distinct demands for which user $\mathcal{U}_{k}$ requests the file $W_n$. The same reasoning follows for each $k \in [K_{\mathcal{U}, b}]$, hence the subfile $W_{n, \mathcal{T}}$ is counted a total of $K_{\mathcal{U}, b} \lambda!t!(\Lambda - \lambda - t)!\binom{\Lambda}{t + \lambda}\binom{N}{K}K!/N $ times within the set of distinct demands for which the file $W_n$ is requested by the users connected to the caches in $\mathcal{U}$ for a given connectivity $b \in \mathcal{B}_{\lambda, \text{b}}$. Now, considering that $\mathcal{B}_{\lambda, \text{b}} = \mathcal{B} \setminus \mathcal{B}_{\lambda, \text{a}}$, it holds that $K_{\mathcal{U}, b} \in [0 : K - 1]$. Moreover, we can easily count how many connectivities $b\in\mathcal{B}_{\lambda, \text{b}}$ exist for which $K_{\mathcal{U}, b} = K - i$ where $i \in [K]$. Indeed, this number of connectivities is equal to the number of non-negative integer solutions of the equation
\begin{equation}
    \sum_{\mathcal{U}' \subseteq [\Lambda]: |\mathcal{U}'| = \lambda, \mathcal{U}' \neq \mathcal{U}} x_{\mathcal{U}'} = i
\end{equation}
which is equal to the number of $\left(\binom{\Lambda}{\lambda} - 1\right)$-weak compositions of the integer $i$. Such number is given by
\begin{equation}
    \binom{i + \binom{\Lambda}{\lambda} - 2}{i}.
\end{equation}
At this point, when we consider all connectivities in $\mathcal{B}_{\lambda, \text{b}}$, there are $\binom{\binom{\Lambda}{\lambda} - 1}{1}$ connectivities with $x_{\mathcal{U}} = K - 1$, there are $\binom{\binom{\Lambda}{\lambda}}{2}$ connectivities with $x_{\mathcal{U}} = K - 2$, there are $\binom{\binom{\Lambda}{\lambda} + 1}{3}$ connectivities with $x_{\mathcal{U}} = K - 3$, and so on. In the end, if we go over all possible connectivities in $\mathcal{B}_{\lambda, \text{b}}$, the subfile $W_{n, \mathcal{T}}$ is counted a total of \begin{equation}
    \lambda!t!(\Lambda - \lambda - t)!\binom{\Lambda}{t + \lambda}\frac{\binom{N}{K}K!}{N} \sum_{i = 1}^{K - 1} (K - i)\binom{i + \binom{\Lambda}{\lambda} - 2}{i}
\end{equation}
times within the set of demands for which the file $W_n$ is requested by users connected to the caches in $\mathcal{U}$. The same reasoning applies for each $\mathcal{U} \subseteq [\Lambda] \setminus \mathcal{T}$ with $|\mathcal{U}| = \lambda$, hence the subfile $W_{n, \mathcal{T}}$ is counted a total of
\begin{equation}
    \lambda!t!(\Lambda - \lambda - t)!\binom{\Lambda}{t + \lambda}\binom{\Lambda - t}{\lambda}\frac{\binom{N}{K}K!}{N} \sum_{i = 1}^{K - 1} (K - i)\binom{i + \binom{\Lambda}{\lambda} - 2}{i}
\end{equation}
times when we focus on connectivities in $\mathcal{B}_{\lambda, \text{b}}$.

Now, if we focus on all connectivities in $\mathcal{B}_\lambda = \mathcal{B}_{\lambda, \text{a}} \cup \mathcal{B}_{\lambda, \text{b}}$, we can see that subfile $W_{n, \mathcal{T}}$ appears in~\eqref{eqn: Complete Lower Bound 2} a total of 
\begin{equation}
    K\binom{\Lambda - t}{\lambda}\Lambda!\frac{\binom{N}{K}K!}{N} + \Lambda! \frac{\binom{\Lambda}{t + \lambda}}{\binom{\Lambda}{t}}\frac{\binom{N}{K}K!}{N} \sum_{i = 1}^{K - 1} (K - i)\binom{i + \binom{\Lambda}{\lambda} - 2}{i}
\end{equation}
times. The same reasoning follows for any $n \in [N]$ and for any $\mathcal{T} \subseteq [\Lambda]$ with $|\mathcal{T}| = t$. Thus, the expression in~\eqref{eqn: Complete Lower Bound 2} can be rewritten after some simplifications as
\begin{align}
    R^\star_{\text{avg}, \mathcal{B}_\lambda} & \geq \sum_{t = 0}^{\Lambda - \lambda + 1} \frac{1}{\left| \mathcal{B}_\lambda \right|}  \left( K\binom{\Lambda - t}{\lambda} + \frac{\binom{\Lambda}{t + \lambda}}{\binom{\Lambda}{t}} \sum_{i = 1}^{K - 1} (K - i)\binom{i + \binom{\Lambda}{\lambda} - 2}{i} \right) x_t \\
                                              & = \sum_{t = 0}^{\Lambda - \lambda + 1} f(t) x_t
\end{align}
where we define
\begin{align}
    f(t) & \coloneqq \frac{1}{\left| \mathcal{B}_\lambda \right|}  \left( K\binom{\Lambda - t}{\lambda} + \frac{\binom{\Lambda}{t + \lambda}}{\binom{\Lambda}{t}} \sum_{i = 1}^{K - 1} (K - i)\binom{i + \binom{\Lambda}{\lambda} - 2}{i} \right) \label{eqn: Objective Function 1} \\
    0 \leq x_t & \coloneqq \sum_{n \in [N]} \sum_{\mathcal{T} \subseteq [\Lambda] : |\mathcal{T}| = t} \frac{\left| W_{n, \mathcal{T}} \right|}{NB}.
\end{align}

At this point, we seek to lower bound the minimum worst-case load $R^\star_{\text{avg}, \mathcal{B}_\lambda}$ by solving the following optimization problem
\begin{subequations}\label{eqn: Optimization Problem 2}
  \begin{alignat}{2}
    & \min_{\bm{x}}  & \quad & \sum_{t = 0}^{\Lambda - \lambda + 1} f(t) x_{t} \\
    & \text{subject to}  &  & \sum_{t = 0}^{\Lambda} x_{t} = 1 \label{eqn: File-Size Constraint 2} \\
    & & & \sum_{t = 0}^{\Lambda} t x_{t} \leq \frac{\Lambda M}{N} \label{eqn: Memory-Size Constraint 2}
  \end{alignat}
\end{subequations}
where \eqref{eqn: File-Size Constraint 2} and \eqref{eqn: Memory-Size Constraint 2} correspond to the file-size constraint and the cumulative cache-size constraint, respectively.

\subsection{Solving the Optimization Problem}

Before proceeding to solve the optimization problem, we simplify the expression $f(t)$ as follows
\begin{align}
    f(t) & = \frac{K}{\left| \mathcal{B}_\lambda \right|} \binom{\Lambda - t}{\lambda} + \frac{1}{\left| \mathcal{B}_\lambda \right|} \frac{\binom{\Lambda}{t + \lambda}}{\binom{\Lambda}{t}} \sum_{i = 1}^{K - 1} (K - i)\binom{i + \binom{\Lambda}{\lambda} - 2}{i} \\
         & = \frac{K}{\left| \mathcal{B}_\lambda \right|} \binom{\Lambda - t}{\lambda} - \frac{K}{\left| \mathcal{B}_\lambda \right|} \frac{\binom{\Lambda}{t + \lambda}}{\binom{\Lambda}{t}} + \frac{1}{\left| \mathcal{B}_\lambda \right|} \frac{\binom{\Lambda}{t + \lambda}}{\binom{\Lambda}{t}} \sum_{i = 0}^{K - 1} (K - i)\binom{i + \binom{\Lambda}{\lambda} - 2}{i} \\
         & = \frac{K}{\left| \mathcal{B}_\lambda \right|} \binom{\Lambda - t}{\lambda} \left( 1 - \frac{\binom{\Lambda}{t + \lambda}}{\binom{\Lambda - t}{\lambda}\binom{\Lambda}{t}} \right) + \frac{1}{\left| \mathcal{B}_\lambda \right|} \frac{\binom{\Lambda}{t + \lambda}}{\binom{\Lambda}{t}} \sum_{i = 0}^{K - 1} (K - i)\binom{i + \binom{\Lambda}{\lambda} - 2}{i} \\
         & = \frac{K}{\left| \mathcal{B}_\lambda \right|} \binom{\Lambda - t}{\lambda} \left( 1 - \frac{1}{\binom{t + \lambda}{\lambda}} \right) + \frac{1}{\left| \mathcal{B}_\lambda \right|} \frac{\binom{\Lambda}{t + \lambda}}{\binom{\Lambda}{t}} \sum_{i = 0}^{K - 1} (K - i)\binom{i + \binom{\Lambda}{\lambda} - 2}{i} \\
         & = A_t + \frac{1}{\left| \mathcal{B}_\lambda \right|} \frac{\binom{\Lambda}{t + \lambda}}{\binom{\Lambda}{t}} \sum_{i = 0}^{K - 1} (K - i)\binom{i + \binom{\Lambda}{\lambda} - 2}{i}
\end{align}
where we define
\begin{equation}
    A_t \coloneqq \frac{K}{\left| \mathcal{B}_\lambda \right|} \binom{\Lambda - t}{\lambda} \left( 1 - \frac{1}{\binom{t + \lambda}{\lambda}} \right).
\end{equation}
Now we can further simplify $f(t)$ as follows
\begin{align}
    f(t) & = A_t + \frac{1}{\left| \mathcal{B}_\lambda \right|} \frac{\binom{\Lambda}{t + \lambda}}{\binom{\Lambda}{t}} \sum_{i = 0}^{K - 1} (K - i)\binom{i + \binom{\Lambda}{\lambda} - 2}{i} \\
         & = A_t + \frac{1}{\left| \mathcal{B}_\lambda \right|}\frac{\binom{\Lambda}{t + \lambda}}{\binom{\Lambda}{t}}  \left(K \sum_{i = 0}^{K - 1} \binom{i + \binom{\Lambda}{\lambda} - 2}{i} - \sum_{i = 0}^{K - 1} i\binom{i + \binom{\Lambda}{\lambda} - 2}{i} \right) \\
         & = A_t + \frac{1}{\left| \mathcal{B}_\lambda \right|}\frac{\binom{\Lambda}{t + \lambda}}{\binom{\Lambda}{t}} \left(K \binom{K + \binom{\Lambda}{\lambda} - 2}{K - 1} - \left( \binom{\Lambda}{\lambda} - 1 \right) \sum_{i = 0}^{K - 2} \binom{i + \binom{\Lambda}{\lambda} - 2}{i} \right) \label{eqn: Application of Hockey-Stick Identity 1} \\
         & = A_t + \frac{1}{\left| \mathcal{B}_\lambda \right|}\frac{\binom{\Lambda}{t + \lambda}}{\binom{\Lambda}{t}} \left(\frac{K^2}{\binom{\Lambda}{\lambda} - 1} \binom{K + \binom{\Lambda}{\lambda} - 2}{K} - \left( \binom{\Lambda}{\lambda} - 1 \right) \binom{K + \binom{\Lambda}{\lambda} - 2}{K - 2} \right) \label{eqn: Application of Hockey-Stick Identity 2} \\
         & = A_t + \frac{1}{\left| \mathcal{B}_\lambda \right|}\frac{\binom{\Lambda}{t + \lambda}}{\binom{\Lambda}{t}} \left(\frac{K^2}{\binom{\Lambda}{\lambda} - 1} \binom{K + \binom{\Lambda}{\lambda} - 2}{K} - \frac{K(K - 1)}{\binom{\Lambda}{\lambda}} \binom{K + \binom{\Lambda}{\lambda} - 2}{K} \right) \\
         & = A_t + \frac{1}{\left| \mathcal{B}_\lambda \right|}\frac{\binom{\Lambda}{t + \lambda}}{\binom{\Lambda}{t}} \frac{K}{\binom{\Lambda}{\lambda}} \binom{K + \binom{\Lambda}{\lambda} - 1}{K} \\
         & = \frac{K\binom{\Lambda}{t + \lambda}}{\binom{\Lambda}{\lambda}\binom{\Lambda}{t}} + A_t
\end{align}
where we used in~\eqref{eqn: Application of Hockey-Stick Identity 1} and in~\eqref{eqn: Application of Hockey-Stick Identity 2} the well-known hockey-stick identity which states that
\begin{equation}
    \sum_{i = k}^{n} \binom{i}{k} = \binom{n + 1}{k + 1}, \quad \forall n, k \in \mathbb{Z}^+, \quad n \geq k.
\end{equation}
At this point, we can rewrite $f(t)$ as
\begin{equation}
    f(t) = \frac{K\binom{\Lambda}{t + \lambda}}{\binom{\Lambda}{\lambda}\binom{\Lambda}{t}} + A_t.
\end{equation}

Now, since the auxiliary variable $x_t$ can be considered once again as a probability mass function, the optimization problem in~\eqref{eqn: Optimization Problem 2} can be seen as the minimization of $\mathbb{E}[f(t)]$. Moreover, the function $f(t)$ can be rewritten also as
\begin{equation}
    f(t) = K\frac{\left| \mathcal{B}_\lambda \right| - \binom{\Lambda}{\lambda}}{\binom{\Lambda}{\lambda}\left| \mathcal{B}_\lambda \right|} \frac{\binom{\Lambda}{t + \lambda}}{\binom{\Lambda}{t}} + \frac{K}{\left| \mathcal{B}_\lambda \right|}\binom{\Lambda - t}{\lambda}.
\end{equation}
Now, from \Cref{lem: Strictly Decreasing Sequence} it follows that the term $\binom{\Lambda}{t + \lambda}/\binom{\Lambda}{t}$ is convex and decreasing in $t$; then, it can be easily checked\footnote{This can be done following the same reasoning presented in~\refappendix{app: Proof of the Strictly Decreasing Sequence Lemma}, after writing down the combinatorial coefficient as a finite product and using the general Leibniz rule to show that its second derivative is non-negative.} that the term $\binom{\Lambda - t}{\lambda}$ is convex and decreasing as well; and finally, it holds that $\left| \mathcal{B}_\lambda \right| \geq \binom{\Lambda}{\lambda}$. Hence, $f(t)$ is a non-negative linear combination of convex and decreasing functions and so it is convex and strictly decreasing for increasing $t$. Consequently, by applying Jensen's inequality we can write $\mathbb{E}[f(t)] \geq f(\mathbb{E}[t]) \geq f(\Lambda M/N)$ after also considering that $\mathbb{E}[t]$ is upper bounded as in~\eqref{eqn: Memory-Size Constraint 2}. Thus, for $t = \Lambda M/N$ and since $K = K'_\lambda$, the optimal average worst-case load $R^\star_{\text{avg}, \mathcal{B}_\lambda}$ is lower bounded by $R_{\text{avg}, \mathcal{B}_\lambda, \text{LB}}$ which is a piece-wise linear curve with corner points
\begin{equation}
    (M, R_{\text{avg}, \mathcal{B}_\lambda, \text{LB}}) = \left(t\frac{N}{\Lambda},  \frac{K'_\lambda \binom{\Lambda}{t + \lambda}}{\binom{\Lambda}{\lambda}\binom{\Lambda}{t}} + A_t \right), \quad \forall t \in [0 : \Lambda - \lambda + 1].
\end{equation}
This concludes the proof.\qed

\section{Proof of \texorpdfstring{\Cref{thm: Converse on General Optimal Average Worst-Case Load}}{Theorem~\ref{thm: Converse on General Optimal Average Worst-Case Load}}}\label{sec: Proof of Converse on General Optimal Average Worst-Case Load}

Since the proof follows the structure of the proof in~\Cref{sec: Proof of Converse on Optimal Average Worst-Case Load}, we will make an effort to avoid repetitions when possible. In this setting, we focus on the ensemble $\mathcal{B}$, which captures all possible connectivities. Here, each of the $K$ users can connect to any set of caches without any specific constraint or structure. Our interest is treating the ensemble as a whole.

\subsection{Constructing the Index Coding Bound}

The mapping of the caching problem to the index coding problem is identical to what we presented in~\Cref{sec: Converse Proof of MACC with Combinatorial Topology} and in~\Cref{sec: Proof of Converse on Optimal Average Worst-Case Load}, hence we directly apply~\Cref{lem: Acyclic Subgraph Converse Bound} using the acyclic set presented in~\Cref{lem: Acyclic Lemma}, obtaining the bound
\begin{equation}\label{eqn: Index Coding Lower Bound 3}
    BR^\star_{b} \geq \sum_{k \in [K_{\emptyset, b}]} \sum_{\mathcal{T} \subseteq [\Lambda]} \left| W_{d_{\emptyset_k, \mathcal{T}}} \right| + \sum_{\lambda \in [\Lambda]} \sum_{i \in [\lambda : \Lambda]} \sum_{\substack{\mathcal{U}^i \subseteq \{c_1, \dots, c_i\} : |\mathcal{U}^i| = \lambda,\\ c_i \in \mathcal{U}^i}} \sum_{k \in [K_{\mathcal{U}^i, b}]} \sum_{\mathcal{T}_i \subseteq [\Lambda] \setminus \{c_1, \dots, c_i\}} \left| W_{d_{\mathcal{U}^i_k, \mathcal{T}_i}} \right|
\end{equation}
for a given connectivity $b \in \mathcal{B}$.

\subsection{Counting the Connectivities}

We are considering the set of connectivities $\mathcal{B}$ and we recall that there are $K$ users. Recall that $K$ is naturally fixed throughout the connectivites, and that consequently it holds that
\begin{equation}
    \sum_{\mathcal{U} \subseteq [\Lambda]} K_{\mathcal{U}, b} = K
\end{equation}
for each $b \in \mathcal{B}$. Toward evaluating the cardinality of the set $\mathcal{B}$, if we let $x_\mathcal{U}$ be a non-negative integer value counting the total number of users connected to the caches in $\mathcal{U}$, then the number of connectitivites in $\mathcal{B}$ is equal to the number of non-negative integer solutions of the equation
\begin{equation}
    \sum_{\mathcal{U} \subseteq [\Lambda]} x_\mathcal{U} = K.
\end{equation}
In this case, this is simply equal to the number of $2^{\Lambda}$-weak compositions~\cite{Stanley2011EnumerativeCombinatorics} of the integer $K$, which is given by
\begin{equation}
    \left| \mathcal{B} \right| = \binom{K + 2^{\Lambda} - 1}{K}.
\end{equation}
Similarly to what we already observed when counting the connectivities in $\mathcal{B}_\lambda$ for the proof of~\Cref{thm: Converse on Optimal Average Worst-Case Load} in~\Cref{sec: Proof of Converse on Optimal Average Worst-Case Load}, now a connectivity $b \in \mathcal{B}$ can be seen as a way to distribute $K$ users among $2^{\Lambda}$ possible states, where each state is an element of the power set of $[\Lambda]$, i.e., a set of $\lambda$ caches --- for some $\lambda \in [0 : \Lambda]$ --- which a user can be connected to.

\subsection{Constructing the Optimization Problem}

Toward building a converse bound on the optimal average worst-case load $R^\star_{\text{avg}, \mathcal{B}}$ over the connectivities in $\mathcal{B}$, we recall that as before we have
\begin{equation}\label{eqn: Average Worst-Case Inequality 2}
    R^\star_{\text{avg}, \mathcal{B}} \geq \frac{1}{\left| \mathcal{B} \right|} \sum_{b \in \mathcal{B}} R^\star_{b}.
\end{equation}

Let us split again the ensemble of interest as $\mathcal{B} = \mathcal{B}_{\text{a}} \cup \mathcal{B}_{\text{b}}$, where $\mathcal{B}_{\text{a}}$ is the set of connectivities for which all the $K$ users are connected to the same set of caches and where $\mathcal{B}_{\text{b}} = \mathcal{B} \setminus \mathcal{B}_{\text{a}}$. For each connectivity $b \in \mathcal{B}_{\text{a}}$, we will create several bounds as the one in~\eqref{eqn: Index Coding Lower Bound 3}. Toward this, we consider any vector $\bm{d} \in \mathcal{D}$ and for each such vector we employ $\Lambda!$ times the permutation vector $\bm{c}_b$ whose first $\lambda$ positions\footnote{Notice that it holds that $\lambda \in [0 : \Lambda]$, since we are considering the most general ensemble of connectivities.} hold the indices of the caches to which the $K$ users are connected according to the connectivity $b \in \mathcal{B}_{\text{a}}$. Also in this case, we assume that these first $\lambda$ positions in $\bm{c}_b$ are put in ascending order and that the same holds, separately, for the remaining $\Lambda - \lambda$ elements. On the other hand, for each connectivity $b \in \mathcal{B}_{\text{b}}$ we will create several lower bounds as the one in~\eqref{eqn: Index Coding Lower Bound 3} by considering any demand vector $\bm{d} \in \mathcal{D}$ and any permutation vector $\bm{c} \in \mathcal{C}$. Our aim is then to take the average of all such bounds to eventually obtain a useful lower bound on the optimal average worst-case load across connectivities in $\mathcal{B}$. Noticing that it holds that $\mathcal{B}_{\text{a}} = \bigcup_{\lambda = 0}^{\Lambda} \mathcal{B}_{\lambda, \text{a}}$ and recalling that we have $|\mathcal{D}| = \binom{N}{K}K!$ and $|\mathcal{C}| = \Lambda!$, we aim to simplify the expression given by
\begin{equation}
    \begin{split}
        \binom{N}{K}K!\Lambda! B \sum_{b \in \mathcal{B}} R^\star_b & \geq \sum_{\lambda = 0}^{\Lambda} \sum_{b \in \mathcal{B}_{\lambda, \text{a}}} \sum_{\bm{d} \in \mathcal{D}} \Lambda! \sum_{k \in [K]} \sum_{\mathcal{T}_i \subseteq [\Lambda] \setminus \{c_{b, 1}, \dots, c_{b, \lambda}\}} \left| W_{d_{\{c_{b, 1}, \dots, c_{b, \lambda}\}_k, \mathcal{T}_i}}\right| + \\
                                                                    & + \sum_{b \in \mathcal{B}_{\text{b}}} \sum_{\bm{d} \in \mathcal{D}} \sum_{\bm{c} \in \mathcal{C}} \Bigg(\sum_{k \in [K_{\emptyset, b}]} \sum_{\mathcal{T} \subseteq [\Lambda]} \left| W_{d_{\emptyset_k, \mathcal{T}}} \right| + \\ 
                                                      & + \sum_{\lambda \in [\Lambda]} \sum_{i \in [\lambda : \Lambda]} \sum_{\substack{\mathcal{U}^i \subseteq \{c_1, \dots, c_i\} : |\mathcal{U}^i| = \lambda, \\ c_i \in \mathcal{U}^i}} \sum_{k \in [K_{\mathcal{U}^i, b}]} \sum_{\mathcal{T}_i \subseteq [\Lambda] \setminus \{c_1, \dots, c_i\}} \left| W_{d_{\mathcal{U}^i_k, \mathcal{T}_i}} \right|\Bigg)  
    \end{split}
\end{equation}
which can be rewritten by means of~\eqref{eqn: Average Worst-Case Inequality 2} as
\begin{equation}\label{eqn: Complete Lower Bound 3}
    \begin{split}
        R^\star_{\text{avg}, \mathcal{B}} & \geq \frac{1}{\binom{N}{K}K!\Lambda! \left| \mathcal{B} \right| B } \Bigg[ \sum_{\lambda = 0}^{\Lambda} \sum_{b \in \mathcal{B}_{\lambda, \text{a}}} \sum_{\bm{d} \in \mathcal{D}} \Lambda! \sum_{k \in [K]} \sum_{\mathcal{T}_i \subseteq [\Lambda] \setminus \{c_{b, 1}, \dots, c_{b, \lambda}\}} \left| W_{d_{\{c_{b, 1}, \dots, c_{b, \lambda}\}_k, \mathcal{T}_i}}\right| + \\
                                                                & + \sum_{b \in \mathcal{B}_{\text{b}}} \sum_{\bm{d} \in \mathcal{D}} \sum_{\bm{c} \in \mathcal{C}} \Bigg(\sum_{k \in [K_{\emptyset, b}]} \sum_{\mathcal{T} \subseteq [\Lambda]} \left| W_{d_{\emptyset_k, \mathcal{T}}} \right| + \\ 
                                                      & + \sum_{\lambda \in [\Lambda]} \sum_{i \in [\lambda : \Lambda]} \sum_{\substack{\mathcal{U}^i \subseteq \{c_1, \dots, c_i\} : |\mathcal{U}^i| = \lambda, \\ c_i \in \mathcal{U}^i}} \sum_{k \in [K_{\mathcal{U}^i, b}]} \sum_{\mathcal{T}_i \subseteq [\Lambda] \setminus \{c_1, \dots, c_i\}} \left| W_{d_{\mathcal{U}^i_k, \mathcal{T}_i}} \right|\Bigg)\Bigg].
    \end{split}
\end{equation}
The next step is to simplify the expression in~\eqref{eqn: Complete Lower Bound 3}. Toward this, we proceed to count how many times each subfile $W_{n, \mathcal{T}}$ --- for any given $n \in [N]$, $\mathcal{T} \subseteq [\Lambda]$ and $|\mathcal{T}| = t$ for some $t \in [0 : \Lambda]$ --- appears in~\eqref{eqn: Complete Lower Bound 3}.

Consider the connectivities in $\mathcal{B}_{\text{a}}$. Assume that the file $W_n$ is demanded by the user $\mathcal{U}_k$ for some $k \in [K]$, for some $\mathcal{U} \subseteq [\Lambda] \setminus \mathcal{T}$ with $|\mathcal{U}| = \lambda$ and for some $\lambda \in [0 : \Lambda - t]$. Since this corresponds to the connectivity $b \in \mathcal{B}_{\text{a}}$ for which all the users are associated to the caches in $\mathcal{U}$, for each demand vector for which the file $W_n$ is requested by user $\mathcal{U}_k$ we employ $\Lambda!$ consecutive times the permutation vector\footnote{For such vector $\bm{c}_b$, we recall that the first $\lambda$ elements correspond to the elements in $\mathcal{U}$. We recall that these first $\lambda$ elements are arranged in ascending order as well as, separately, the remaining $\Lambda - \lambda$ elements.} $\bm{c}_{b}$. Since here, by construction, the elements in $\mathcal{T}$ are after the elements in $\mathcal{U}$ in the vector $\bm{c}_{b}$, we can deduce that the subfile $W_{n, \mathcal{T}}$ is counted $\Lambda!$ times --- namely, once for each of the $\Lambda!$ times that the vector $\bm{c}_{b}$ is employed --- whenever the file $W_n$ is requested by user $\mathcal{U}_k$. Recalling that any file is requested by user $\mathcal{U}_{k}$ a total of $\binom{N}{K}K!/N$ times, the subfile $W_{n, \mathcal{T}}$ is counted $\Lambda!\binom{N}{K}K!/N$ within the set of demands for which user $\mathcal{U}_k$ requests the file $W_n$. The same reasoning follows for each $k \in [K]$, for each $\mathcal{U} \subseteq [\Lambda] \setminus \mathcal{T}$ with $|\mathcal{U}| = \lambda$ and for each $\lambda \in [0 : \Lambda - t]$, hence the subfile $W_{n, \mathcal{T}}$ is counted
\begin{equation}
    \sum_{\lambda = 0}^{\Lambda - t} K\binom{\Lambda - t}{\lambda}\Lambda!\frac{\binom{N}{K}K!}{N}
\end{equation}
times across the set of connectivities in $\mathcal{B}_{\text{a}}$.

Let us now consider a specific connectivity $b \in \mathcal{B}_{\text{b}}$. Assume that the file $W_n$ is demanded by the user $\mathcal{U}_{k}$ for some $k \in [K_{\mathcal{U}, b}]$, for some $\mathcal{U} \subseteq [\Lambda] \setminus \mathcal{T}$ with $|\mathcal{U}| = \lambda$ and for some $\lambda \in [0 : \Lambda - t]$, recalling once again that such file is requested by user $\mathcal{U}_{k}$ a total of $\binom{N}{K}K!/N$ times. Within the set of demands for which user $\mathcal{U}_{k}$ requests the file $W_n$, the subfile $W_{n, \mathcal{T}}$ is counted only when the elements in the set $\mathcal{U}$ appear in the vector $\bm{c}$ before the elements in the set $\mathcal{T}$. Since there is a total of $\lambda!t!(\Lambda - \lambda - t)!\binom{\Lambda}{t + \lambda}$ such vectors $\bm{c}$, the subfile $W_{n, \mathcal{T}}$ is counted a total of $\lambda!t!(\Lambda - \lambda - t)!\binom{\Lambda}{t + \lambda}\binom{N}{K}K!/N$ times within the set of distinct demands for which user $\mathcal{U}_{k}$ requests the file $W_n$. The same reasoning follows for each $k \in [K_{\mathcal{U}, b}]$, hence the subfile $W_{n, \mathcal{T}}$ is counted a total of $K_{\mathcal{U}, b} \lambda!t!(\Lambda - \lambda - t)!\binom{\Lambda}{t + \lambda}\binom{N}{K}K!/N $ times within the set of distinct demands for which the file $W_n$ is requested by the users connected to the caches in $\mathcal{U}$ for a given connectivity $b \in \mathcal{B}$. Now, recalling that $\mathcal{B}_{\text{b}} = \mathcal{B} \setminus \mathcal{B}_{\text{a}}$, it holds that $K_{\mathcal{U}, b} \in [0 : K - 1]$. Moreover, we can easily count how many connectivities $b \in \mathcal{B}_{\text{b}}$ exist with $K_{\mathcal{U}, b} = K - i$ where $i \in [K]$. Indeed, this number of connectivities is equal to the number of non-negative integer solutions of the equation
\begin{equation}
    \sum_{\mathcal{U}' \subseteq [\Lambda]: \mathcal{U}' \neq \mathcal{U}} x_{\mathcal{U}'} = i
\end{equation}
which is equal to the number of $\left(2^{\Lambda} - 1\right)$-weak compositions of the integer $i$. Such number is given by
\begin{equation}
    \binom{i + 2^{\Lambda} - 2}{i}.
\end{equation}
Hence, going over all connectivities in $\mathcal{B}_{\text{b}}$, the subfile $W_{n, \mathcal{T}}$ is counted a total of \begin{equation}
    \lambda!t!(\Lambda - \lambda - t)!\binom{\Lambda}{t + \lambda}\frac{\binom{N}{K}K!}{N} \sum_{i = 1}^{K - 1} (K - i)\binom{i + 2^{\Lambda} - 2}{i}
\end{equation}
times within the set of distinct demands for which the file $W_n$ is requested by users connected to the $\lambda$ caches in $\mathcal{U}$. The same reasoning applies for each $\mathcal{U} \subseteq [\Lambda] \setminus \mathcal{T}$ with $|\mathcal{U}| = \lambda$ and for each $\lambda \in [0 : \Lambda - t]$, and thus the subfile $W_{n, \mathcal{T}}$ is counted a total of
\begin{equation}
    \sum_{\lambda = 0}^{\Lambda - t} \lambda!t!(\Lambda - \lambda - t)!\binom{\Lambda}{t + \lambda}\binom{\Lambda - t}{\lambda}\frac{\binom{N}{K}K!}{N} \sum_{i = 1}^{K - 1} (K - i)\binom{i + 2^{\Lambda} - 2}{i}
\end{equation}
times when we span across connectivities in $\mathcal{B}_{\text{b}}$.

Now, going through all connectivities in $\mathcal{B} = \mathcal{B}_{\text{a}} \cup \mathcal{B}_{\text{b}}$, we can see that the subfile $W_{n, \mathcal{T}}$ appears in~\eqref{eqn: Complete Lower Bound 3} a total of 
\begin{equation}
    \sum_{\lambda = 0}^{\Lambda - t} K\binom{\Lambda - t}{\lambda}\Lambda!\frac{\binom{N}{K}K!}{N} + \Lambda! \frac{\binom{\Lambda}{t + \lambda}}{\binom{\Lambda}{t}} \frac{\binom{N}{K}K!}{N} \sum_{i = 1}^{K - 1} (K - i)\binom{i + 2^{\Lambda} - 2}{i}
\end{equation}
times. The same reasoning follows for any $n \in [N]$ and for any $\mathcal{T} \subseteq [\Lambda]$ with $|\mathcal{T}| = t$. Thus, the expression in~\eqref{eqn: Complete Lower Bound 3} can be rewritten after some simplifications as
\begin{align}
    R^\star_{\text{avg}, \mathcal{B}} & \geq \sum_{t = 0}^{\Lambda} \sum_{\lambda = 0}^{\Lambda - t} \frac{1}{|\mathcal{B}|} \left( K\binom{\Lambda - t}{\lambda} + \frac{\binom{\Lambda}{t + \lambda}}{\binom{\Lambda}{t}} \sum_{i = 1}^{K - 1} (K - i)\binom{i + 2^{\Lambda} - 2}{i} \right) x_t \\
                                      & = \sum_{t = 0}^{\Lambda} f(t) x_t
\end{align}
where we define
\begin{align}
    f(t) & \coloneqq \sum_{\lambda = 0}^{\Lambda - t} \frac{1}{|\mathcal{B}|} \left( K\binom{\Lambda - t}{\lambda} + \frac{\binom{\Lambda}{t + \lambda}}{\binom{\Lambda}{t}} \sum_{i = 1}^{K - 1} (K - i)\binom{i + 2^{\Lambda} - 2}{i} \right) \label{eqn: Objective Function 2} \\
    0 \leq x_t & \coloneqq \sum_{n \in [N]} \sum_{\mathcal{T} \subseteq [\Lambda] : |\mathcal{T}| = t} \frac{\left| W_{n, \mathcal{T}} \right|}{NB}.
\end{align}

At this point, formulating the optimization problem as in~\eqref{eqn: Optimization Problem 1} and as in~\eqref{eqn: Optimization Problem 2}, we seek to solve the following
\begin{subequations}
  \begin{alignat}{2}
    & \min_{\bm{x}}  & \quad & \sum_{t = 0}^{\Lambda} f(t) x_{t} \\
    & \text{subject to}  &  & \sum_{t = 0}^{\Lambda} x_{t} = 1 \\
    & & & \sum_{t = 0}^{\Lambda} t x_{t} \leq \frac{\Lambda M}{N}.
  \end{alignat}
\end{subequations}

\subsection{Solving the Optimization Problem}

Before proceeding to solve the optimization problem, we aim to simplify the expression $f(t)$. Toward this, we notice that the only difference between the expression in~\eqref{eqn: Objective Function 1} and the expression in~\eqref{eqn: Objective Function 2} is that the latter is a summation over $\lambda$ of $\Lambda - t + 1$ terms, where each of these terms is the same as the expression in~\eqref{eqn: Objective Function 1} except for the considered ensemble $\mathcal{B}$ and for the term $2^{\Lambda}$ in place of $\binom{\Lambda}{\lambda}$ in the binomial coefficient which appears in the summation over $i$. This latter difference does not change the simplification steps presented in the previous proof for what concerns each summand in~\eqref{eqn: Objective Function 2}, hence we can directly conclude that $f(t)$ can be simplified as
\begin{align}
    f(t) & = \sum_{\lambda = 0}^{\Lambda - t}\frac{K}{| \mathcal{B} |} \binom{\Lambda - t}{\lambda} + \frac{1}{| \mathcal{B} |} \frac{\binom{\Lambda}{t + \lambda}}{\binom{\Lambda}{t}} \sum_{i = 1}^{K - 1} (K - i)\binom{i + 2^{\Lambda} - 2}{i} \\
         & =\sum_{\lambda = 0}^{\Lambda - t} \frac{K\binom{\Lambda}{t + \lambda}}{2^{\Lambda}\binom{\Lambda}{t}} + A_{t, \lambda}
\end{align}
where we define
\begin{equation}
    A_{t, \lambda} \coloneqq \frac{K}{| \mathcal{B} |} \binom{\Lambda - t}{\lambda} \left( 1 - \frac{1}{\binom{t + \lambda}{\lambda}} \right).
\end{equation}

Now, it can be shown that each summand 
\begin{equation}
    \frac{K\binom{\Lambda}{t + \lambda}}{2^{\Lambda}\binom{\Lambda}{t}} + A_{t, \lambda}
\end{equation}
is convex and decrasing for fixed value of $\lambda$, as we did previously in~\Cref{sec: Proof of Converse on Optimal Average Worst-Case Load}. Hence, we can conclude also in this case that the function $f(t)$ is a non-negative linear combination of convex and decreasing functions, and thus we can solve the optimization problem again applying Jensen's inequality and using the memory-size constraint. Consequently, for $t = \Lambda M/N$, following the same steps as in the proof of~\Cref{sec: Proof of Converse on Optimal Average Worst-Case Load}, we can now conclude that the optimal average worst-case load $\mathcal{B}$ is lower bounded by $R_{\text{avg}, \mathcal{B}, \text{LB}}$ which is a piece-wise linear curve with corner points
\begin{equation}
    (M, R_{\text{avg}, \mathcal{B}, \text{LB}}) = \left(t \frac{N}{\Lambda}, \sum_{\lambda = 0}^{\Lambda - t} \frac{K \binom{\Lambda}{t + \lambda}}{2^{\Lambda}\binom{\Lambda}{t}} + A_{t, \lambda} \right), \quad \forall t \in [0 : \Lambda]
\end{equation}
which concludes the proof.\qed

\section{Conclusions}\label{sec: Conclusions}

In this work, we identified the exact fundamental limits of multi-access caching for important general topologies as well as have derived information-theoretic converses that bound the performance across ensembles of connectivities.

Our first contribution was to derive the fundamental limits of multi-access caching for the here-introduced generalized combinatorial topology that stands out for the unprecedented coding gains that it allows. These gains --- first recorded in~\cite{Muralidhar2021Maddah-Ali-NiesenCaching} for a specific instance of this generalized topology --- are proven in this work to be optimal, and to hold not only for a much denser range of users, but also in the presence of a coexistence of users connected to different numbers of caches. As a direct consequence of identifying the exact fundamental limits of the generalized combinatorial topology, we now know that a basic and fixed MAN cache placement can \emph{optimally} handle any generalized combinatorial network regardless of having unknown numbers of users with unequal cache-connectivity capabilities.

Subsequently, we considered the optimal performance of different ensembles of connectivities, deriving a lower bound on the average performance for a large ensemble of interest as well as for the entire ensemble of all possible topologies. To the best of our knowledge, this is the first work that explores the fundamental limits of the average performance across ensembles of connectivities as well as the first work to consider the topology-agnostic multi-access setting. We hope this contribution applies toward better understanding the role of topology in defining the performance as well as the corresponding role of topology-awareness.
\begin{itemize}
    \item \emph{Practical Pertinence and Optimality of the Generalized Combinatorial Topology}. As the subpacketization constraint is one of the main bottlenecks limiting the actual use of coded caching in various real applications, the combinatorial topology is undoubtedly a promising expedient to overcome this bottleneck. Indeed, this multi-access approach allows to serve many users at a time while controlling the subpacketization parameter by carefully calibrating the number of caches $\Lambda$. This particular topology breaks the barrier of having coding gains that are close to $\Lambda\gamma + 1$, and it goes one step further by allowing the unprecedented gains that are now a polynomial power of the cache redundancy, to be achieved with a modest subpacketization, a reasonable amount of cache resources and a very modest connectivity investment. Indeed, the gains explode even when $\lambda$ is as small as $\lambda = 2$ and this increase --- we re-emphasize --- is without any increase in the caching resources $\Lambda\gamma$. Our generalized model extends even further the benefits of this topology, allowing for a denser range of users $K$ and supporting the coexistence of users connected to different numbers of caches. As a side-product of here identifying the fundamental limits of performance in this broad topology, our result shows that, when there are users connected to a different number of caches $\lambda$ with $\lambda \in [0 : \Lambda]$ in accordance to the combinatorial topology, then a MAN cache placement and a simple TDMA-like application of the scheme in~\cite{Muralidhar2021Maddah-Ali-NiesenCaching} is enough to achieve the minimum possible load in the worst-case scenario. This implies that treating each $\lambda$-setting separately is optimal and so there would not be any advantage in encoding across users connected to a different number of caches.
    \item \emph{An Agnostic Perspective and the Search for the Best Topology}. In real scenarios the actual connectivity could change over time. In this context, under the assumption that the cache placement procedure is performed only once regardless of the connectivity, the results in~\Cref{thm: Converse on Optimal Average Worst-Case Load} and in~\Cref{thm: Converse on General Optimal Average Worst-Case Load} are of particular interest since they lower bound the optimal average worst-case load in the agnostic scenario. Furthermore, interestingly, a closer inspection reveals that~\Cref{thm: Converse on Optimal Average Worst-Case Load} holds also under the standard MAN cache placement\footnote{The converse in~\Cref{thm: Converse on Optimal Average Worst-Case Load} holds under a generic uncoded cache placement, and the final expression is a result of optimizing over all fixed placements (and over all dynamic delivery schemes). However, a closer inspection reveals that, instead of assuming a generic uncoded placement as in~\Cref{sec: Constructing the Index Coding Bound}, the placement can alternatively be forced to be the one presented in~\Cref{sec: Achievability Proof of MACC with Generalized Combinatorial Topology}, and the resulting lower bound would be the same.} described in~\Cref{sec: Achievability Proof of MACC with Generalized Combinatorial Topology}. This brings to the fore the open question of whether, under the assumption of a MAN placement, the combinatorial topology is indeed the best possible topology. This question is supported by the just-now mentioned specific optimization solution that leads to~\Cref{thm: Converse on Optimal Average Worst-Case Load}, as well as by the large gains that the MAN placement achieves over this topology, irrespective of the fact that different users may be connected to a different number of caches.
\end{itemize}

This work is part of a sequence of works that suggest the importance of the multi-access caching paradigm not only for its role in boosting the gains in coded caching, which had remained relatively low due to the severe subpacketization bottleneck, but also for its role in several distributed communication paradigms where communication complexity is traded off with computational complexity (see \cite{Wan2021OnComputation} and other related works \cite{Ye2018Communication-ComputationCoding, Dutta2019Short-Dot:Products, Elmahdy2020OnLearning}) as a function of how servers are connected to data sources. Finally, while most works on the MACC focused mainly on the cyclic wrap-around topology, our work further legitimizes the effort of investigating connectivities different from the originally proposed MACC model. The bounds on the average performance over the ensemble of connectivities leave open the possibility that there may exist another topology performing uniformly better than the powerful optimal performance of the generalized combinatorial topology identified in \Cref{thm: MACC with Generalized Combinatorial Topology}. Finding such a topology, if it exists, would be indeed an exciting proposition given the already large gains that are associated to the combinatorial topology.

\appendices

\section{Proof of \texorpdfstring{\Cref{lem: Acyclic Lemma}}{Lemma~\ref{lem: Acyclic Lemma}}}\label{app: Proof of Acyclic Lemma}

We show how the set described in~\Cref{lem: Acyclic Lemma} is guaranteed to be acyclic. Given any connectivity $b \in \mathcal{B}$, let us start by considering the set
\begin{equation}\label{eqn: First Acyclic Set}
    \bigcup_{k \in [K_{\emptyset, b}]} \bigcup_{\mathcal{T} \subseteq [\Lambda]} \left\{ W_{d_{\emptyset_k, \mathcal{T}}} \right\}.
\end{equation}
Since each user $\emptyset_k$ with $k \in [K_{\emptyset, b}]$ is not connected to any cache, we can conclude that there is no connection among the vertices corresponding to the subfiles $W_{d_{\emptyset_k, \mathcal{T}}}$ for each $k \in [K_{\emptyset, b}]$ and $\mathcal{T} \subseteq [\Lambda]$. This directly means that the set in~\eqref{eqn: First Acyclic Set} is acyclic. Now consider the set of vertices
\begin{equation}\label{eqn: Second Acyclic Set}
    \bigcup_{\lambda \in [\Lambda]}  \bigcup_{i \in [\lambda : \Lambda]}  \bigcup_{\substack{\mathcal{U}^i \subseteq \{c_1, \dots, c_i\} : |\mathcal{U}^i| = \lambda,\\ c_i \in \mathcal{U}^i}}  \bigcup_{k \in [K_{\mathcal{U}^i, b}]}  \bigcup_{\mathcal{T}_i \subseteq [\Lambda] \setminus \{c_1, \dots, c_i\}} \left\{ W_{d_{\mathcal{U}^i_k, \mathcal{T}_i}} \right\}.
\end{equation}
Following the same reasoning as in the proof of~\cite[Lemma 1]{Wan2020AnPlacement}, for a specific permutation of caches $\bm{c} = (c_1, \dots, c_\Lambda)$ we will say that the subfile $W_{d_{\mathcal{U}^i_k, \mathcal{T}_i}}$ belongs to the $i$-th \emph{level}\footnote{The term \emph{level} carries the same meaning as in the proof of~\cite[Lemma 1]{Wan2020AnPlacement}, and its impact here is described mathematically in compliance with our setting.}, which will mean that $\mathcal{U}^i \subseteq \{c_1, \dots, c_i\}$ with $|\mathcal{U}^i| = \lambda$ and $c_i \in \mathcal{U}^i$, and that $\mathcal{T}_i \subseteq [\Lambda] \setminus \{c_1, \dots, c_i\}$ with $i \in [\lambda : \Lambda]$. As one may see, no user in the $i$-th level has access to the subfiles in its level, since $\mathcal{U}^i \cap \mathcal{T}_i = \emptyset$ for each $\mathcal{U}^i \subseteq \{c_1, \dots, c_i\}$ with $|\mathcal{U}^i| = \lambda$ and $c_i \in \mathcal{U}^i$, and for each $\mathcal{T}_i \subseteq [\Lambda] \setminus \{c_1, \dots, c_i\}$. Moreover, no user in the $i$-th level has access to the subfiles in higher levels, since $\mathcal{U}^i \cap \mathcal{T}_j = \emptyset$ for each $\mathcal{U}^i \subseteq \{c_1, \dots, c_i\}$ with $|\mathcal{U}^i| = \lambda$ and $c_i \in \mathcal{U}^i$, and for each $\mathcal{T}_j \subseteq [\Lambda] \setminus \{c_1, \dots, c_j\}$ with $j \in [i + 1 : \Lambda]$. Consequently, we can conclude that the set in \eqref{eqn: Second Acyclic Set} is acyclic. Moreover, since by definition no user $\emptyset_k$ with $k \in [K_{\emptyset, b}]$ is connected to any cache, there cannot be any cycle between the set in \eqref{eqn: First Acyclic Set} and the set in \eqref{eqn: Second Acyclic Set}, and thus the union set
\begin{equation}
    \bigcup_{k \in [K_{\emptyset, b}]} \bigcup_{\mathcal{T} \subseteq [\Lambda]} \left\{ W_{d_{\emptyset_k, \mathcal{T}}} \right\} \cup \bigcup_{\lambda \in [\Lambda]}  \bigcup_{i \in [\lambda : \Lambda]}  \bigcup_{\substack{\mathcal{U}^i \subseteq \{c_1, \dots, c_i\} : |\mathcal{U}^i| = \lambda,\\ c_i \in \mathcal{U}^i}}  \bigcup_{k \in [K_{\mathcal{U}^i, b}]}  \bigcup_{\mathcal{T}_i \subseteq [\Lambda] \setminus \{c_1, \dots, c_i\}} \left\{ W_{d_{\mathcal{U}^i_k, \mathcal{T}_i}} \right\}
\end{equation}
is also acyclic. This concludes the proof.\qed

\section{Proof of \texorpdfstring{\Cref{lem: Strictly Decreasing Sequence}}{Lemma~\ref{lem: Strictly Decreasing Sequence}}}\label{app: Proof of the Strictly Decreasing Sequence Lemma}

Let us write $f(t) = \sum_{\lambda = 0}^{\Lambda - t} f_\lambda(t) = K_0 + \sum_{\lambda = 1}^{\Lambda} f_\lambda(t)$, where we define
\begin{equation}
    f_\lambda(t) \coloneqq K_\lambda \frac{\binom{\Lambda}{t + \lambda}}{\binom{\Lambda}{t}}
\end{equation}
with $f_\lambda(t) = 0$ if $t > \Lambda - \lambda$. Regarding convexity, we will prove each $f_\lambda(t)$ to be convex for $t \in [0 : \Lambda]$, which will then tell us that $f(t)$ is also convex, since a non-negative linear combination of functions that are convex on the same interval is also convex. Notice that, since any linear function is both convex and concave, the constant term $K_0$ is convex with respect to $t$.

When showing the convexity of $f_\lambda(t)$, we can focus on the range $t \leq \Lambda - \lambda + 1$ since we know that $f_\lambda(t) = 0$ for any $t \geq \Lambda - \lambda + 1$. Doing so will automatically allow us to conclude that this same function $f_\lambda(t)$ is also convex over the entire\footnote{Indeed, if the function $f_\lambda(t)$ is convex in the interval $t \in [0 : \Lambda - \lambda + 1]$ and $f_\lambda(t) = 0$ for $t \in [\Lambda - \lambda + 1 : \Lambda]$ --- which implies also convexity in the interval $t \in [\Lambda - \lambda + 1 : \Lambda]$ since any linear function is both convex and concave --- then $f_\lambda(t)$ is convex also over the entire interval $t \in [0 : \Lambda]$, since it should be self-evident that the line segment between any two points $t_1, t_2$ with $t_1 \in [0 : \Lambda - \lambda + 1]$ and $t_2 \in [\Lambda - \lambda + 1 : \Lambda]$ lies above the graph between such two points.} interval $t \in [0 : \Lambda]$. The convexity of $f_\lambda(t)$ can be easily shown by verifying that the second derivative $f^{(2)}_\lambda(t)$ with respect to $t$ is non-negative for $t \in [0 : \Lambda - \lambda + 1]$. To see this, let us first note that 
\begin{align}
     f_\lambda(t) & = K_\lambda \frac{\binom{\Lambda}{t + \lambda}}{\binom{\Lambda}{t}} \\
                  & =  K_\lambda \prod_{i = 0}^{\lambda - 1} \frac{\Lambda - t - i}{t + \lambda - i} \\
                  & =  K_\lambda \prod_{i = 0}^{\lambda - 1} f_{\lambda, i}(t)
\end{align}
where $f_{\lambda, i}(t) \coloneqq (\Lambda - t - i)/(t + \lambda - i)$. Applying the general Leibniz rule allows us to rewrite the second derivative $f^{(2)}_\lambda(t)$ as
\begin{align}
    f^{(2)}_\lambda(t) & = \left(K_\lambda \prod_{i = 0}^{\lambda - 1} f_{\lambda, i}(t) \right)^{(2)} \\
                       & = K_\lambda \sum_{k_0 + \dots + k_{\lambda - 1} = 2} \binom{2}{k_0, \dots, k_{\lambda - 1}} \prod_{i = 0}^{\lambda - 1} f_{\lambda, i}^{(k_i)}(t)\label{eqn: Second Derivative}
\end{align}
where $k_\ell \in \mathbb{Z}^{+}_{0}$ for each $\ell \in [0 : \lambda - 1]$. Now we will prove that $f^{(2)}_\lambda(t) \geq 0$ by showing that $\prod_{i = 0}^{\lambda - 1} f_{\lambda, i}^{(k_i)}(t) \geq 0$ for any summand in~\eqref{eqn: Second Derivative}, noticing that we have a summand for every $\lambda$-weak composition of $2$. Such weak compositions can be divided into two classes.
\begin{itemize}
    \item The first class includes weak compositions where $k_m = 2$ for some $m \in [0 : \lambda - 1]$ and $k_n = 0$ for each $n \in [0 : \lambda - 1] \setminus \{m\}$. In this case, the term $\prod_{i = 0}^{\lambda - 1} f_{\lambda, i}^{(k_i)}(t)$ is given by the product of $\lambda - 1$ functions $f_{\lambda, n}^{(0)}(t)$ for $n \in [0 : \lambda - 1] \setminus \{m\}$ with $1$ additional function $f_{\lambda, m}^{(2)}(t)$, where $f_{\lambda, n}^{(0)}(t)$ is simply the function $f_{\lambda, n}(t)$ and where $f_{\lambda, m}^{(2)}(t)$ is the second derivative of $f_{\lambda, m}(t)$. Hence, if $f_{\lambda, n}^{(0)}(t) \geq 0$ and $f_{\lambda, m}^{(2)}(t) > 0$, then $\prod_{i = 0}^{\lambda - 1} f_{\lambda, i}^{(k_i)}(t) \geq 0$.
    \item The second class includes weak compositions where $k_{m_1} = k_{m_2} = 1$ for some $m_1, m_2 \in [0 : \lambda - 1]$ with $m_1 \neq m_2$ and $k_n = 0$ for each $n \in [0 : \lambda - 1] \setminus \{m_1, m_2\}$. In this case, the term $\prod_{i = 0}^{\lambda - 1} f_{\lambda, i}^{(k_i)}(t)$ is given by the product of $\lambda - 2$ functions $f_{\lambda, n}^{(0)}(t)$ for $n \in [0 : \lambda - 1] \setminus \{m_1, m_2\}$ with $2$ additional functions $f_{\lambda, m_1}^{(1)}(t)$ and $f_{\lambda, m_2}^{(1)}(t)$, where now $f_{\lambda, n}^{(1)}(t)$ is the first derivative of $f_{\lambda, n}(t)$. Hence, if $f_{\lambda, n}^{(0)}(t) \geq 0$, $f_{\lambda, m_1}^{(1)}(t) < 0$ and $f_{\lambda, m_2}^{(1)}(t) < 0$, then $\prod_{i = 0}^{\lambda - 1} f_{\lambda, i}^{(k_i)}(t) \geq 0$.
\end{itemize}
At this point, in order to prove that $\prod_{i = 0}^{\lambda - 1} f_{\lambda, i}^{(k_i)}(t) \geq 0$ for any $\lambda$-weak composition of $2$, we have to prove that the following inequalities
\begin{align}
    f^{(0)}_{\lambda, i}(t) & \geq 0 \\
    f^{(1)}_{\lambda, i}(t) & < 0 \\
    f^{(2)}_{\lambda, i}(t) & > 0
\end{align}
hold for each $i \in [0 : \lambda - 1]$, taking into account that 
\begin{align}
    f^{(1)}_{\lambda, i}(t) & = \frac{2i - \Lambda - \lambda}{(t + \lambda - i)^2} \\
    f^{(2)}_{\lambda, i}(t) & = -\frac{2}{t + \lambda - i}f^{(1)}_{\lambda, i}(t).
\end{align}
Since $i \leq \lambda - 1$ and $\lambda \leq \Lambda - t + 1$, it holds that $f^{(0)}_{\lambda, i}(t) \geq 0$. Moreover, it holds that 
\begin{align}
    2i - \Lambda - \lambda & \leq 2(\lambda - 1) - \Lambda - \lambda \\
                           & = \lambda - 2 - \Lambda \\
                           & \leq -2 \\
                           & < 0.
\end{align}
Consequently, we can conclude that $f^{(1)}_{\lambda, i}(t) < 0$, since the numerator is always negative and the denominator is always positive. We can also conclude that $f^{(2)}_{\lambda, i}(t) > 0$, since $f^{(2)}_{\lambda, i}(t)$ is equal to the first derivative $f^{(1)}_{\lambda, i}(t) < 0$ multiplied by a strictly negative term $-2/(t + \lambda - i)$. Thus, given that the second derivative $f^{(2)}_\lambda(t)$ equals the sum of non-negative terms, we can conclude that $f^{(2)}_\lambda(t) \geq 0$ and that $f_\lambda(t)$ is convex for each $\lambda \in [\Lambda]$ and $t \in [0 : \Lambda - \lambda + 1]$. As a consequence, $f(t)$ is convex over the entire range of interest.

The fact that $f(t)$ is decreasing in $t \in [0 : \Lambda]$ follows from the fact that $f(0) = K\geq f(\Lambda) = K_0$ and from the aforementioned convexity of $f(t)$. This concludes the proof.\qed

\printbibliography

\end{document}